\documentclass[11pt]{iopart}
\usepackage{iopams}  

\usepackage{graphicx}
\usepackage{xcolor}
\usepackage{url}
\usepackage{times}
\usepackage{bbm}
\usepackage{amssymb,amsthm,amsfonts,amstext,color}
\usepackage[english]{babel}
\usepackage[colorlinks=true,breaklinks=true]{hyperref}

\usepackage{paralist} 

\definecolor{jens}{rgb}{0,0.8,0.4}
\newcommand{\je}[1]{{\color{black} #1}}
\newcommand{\jen}[1]{{\color{black} #1}}
\newcommand{\rodrigo}[1]{{\color{black} #1}}
\newcommand{\arnau}[1]{{\color{black} #1}}

\newcommand\EE{{\mathbbm{E}}}

\newcommand\id{{\mathbbm{1}}}

\newcommand{\Trace}{{\rm tr}} 

\newcommand{\bra}[1]{\langle #1|}
\newcommand{\ket}[1]{|#1\rangle}
\newcommand{\braket}[2]{\langle #1|#2\rangle}
\newcommand{\ketbra}[2]{| #1 \rangle \langle #2 |}

\newcommand{\proj}[1]{\vert #1\rangle\!\langle#1 \vert}

\newtheorem{thm}{Theorem}

\newtheorem{assum}{Assumption}
\newtheorem{cor}{Corollary}

\begin{document}

\title{
Thermal machines beyond the weak coupling regime
}  

\author{R. Gallego$^1$, A. Riera$^{1,2}$, and J. Eisert$^1$} 
\address{1 Dahlem Center for Complex Quantum Systems, Freie Universit\"at Berlin, 14195 Berlin, Germany}
\address{2 ICFO-Institut de Ciencies Fotoniques, Mediterranean Technology Park, 08860 Castelldefels, Spain}

\begin{abstract}
How much work can be extracted from a heat bath using a thermal machine? The study of this question has a very long tradition in statistical physics
in the weak-coupling limit, applied to macroscopic systems. 
However, the assumption that thermal heat baths remain uncorrelated with physical 
systems at hand is less reasonable on the nano-scale and in the quantum setting.
In this work, we establish a  framework of work extraction in the presence of quantum correlations. We show in a mathematically rigorous and quantitative
fashion that quantum correlations and entanglement emerge as a limitation to work extraction compared to what would be allowed by the second law 
of thermodynamics. At the heart of the approach are operations that capture naturally non-equilibrium dynamics encountered when putting physical  systems into contact with 
each other. We discuss various limits that relate to known results and put our work into context of 
approaches to finite-time quantum thermodynamics.
\end{abstract}


\section{Introduction}
The theory of thermodynamics originates from the study of thermal machines in the early industrial age, when 
it was of utmost importance to find out what rates of work extraction could ultimately be achieved. Early on, it became
clear that the theory of thermal machines would be intimately related with topics of fundamental 
physics such as statistical mechanics and notions of classical information theory \cite{Giles}.
Here, the interplay and relations between widely-studied notions of
 \emph{work}, \emph{entropy} and of \emph{statistical ensembles} are in the focus of attention. Concomitant with the technological development,
 the theory also became more intricate and addressed more elaborate situations.
Famous thought experiments such as Maxwell's demon, Landauer's erasure, and Slizard's engine have not only
puzzled researchers for a long time, but today also serve as a source of inspiration for quantitative studies of achievable rates
when employing thermal machines \cite{Alicki,Kosloff2,Bennett,Henrich}. Indeed, with nano-machines operating at or close to the quantum level coming into reach, there has
recently been an explosion of interest on the question what role quantum effects may possibly play. It is the potential and
limits
of work extraction with physically plausible operations which respect quantum correlations that are established in this work.

The role of correlations is already a challenging problem with a long tradition in classical thermodynamics. 
Thermal machines comprises a system that is brought into contact with a thermal bath. 
This process introduces correlations that are typically disregarded by assuming that the interaction between system and bath is sufficiently weak.
Due to the limited applicability of this assumption in practical situations, there has been a great effort in characterising thermal machines beyond the weak-coupling regime in specific models 
both in the classical and quantum setting \cite{Allahverdyan2000,Allahverdyan2001,Hilt2011,Campisi2010}. 
However, a general framework for work-extraction beyond the weak coupling assumption is still missing. 
This is mainly due to a lack of understanding of the process of evolution towards equilibrium under the effect of generic strong couplings, which has only began to be tackled in full extent in the last years.

More specifically, \arnau{let us introduce the \emph{weak-coupling assumption} \jen{precisely} as}: A system $S$ with Hamiltonian $H_S$, when put into weak thermal contact 
with a thermal bath $B$, equilibrates towards the state
\begin{equation}\label{eq:Gibbstandard}
\rho_{S}=\omega(H_S)
\end{equation}
with $\omega(H_S):=-e^{\beta H_{S}}/Z$, $Z:= \tr(-\beta H_{S})$ and $\beta>0$ being the inverse temperature. 
That is, $S$ equilibrates to the usual Gibbs ensemble. 
\arnau{Note that this notion of weak-coupling can in general differ from the one \jen{sometimes used in the study of} open quantum systems leading
to Markovian dynamics of the subsystem $S$ \cite{Davies,Kosloff}.}
The precise conditions on the coupling so that (\ref{eq:Gibbstandard}) is fulfilled have been recently tackled in the quantum setting: the strength of the coupling Hamiltonian $V$ -- measured in an adequate norm
-- has to be negligible in comparison with the intensive 
thermal energy scale $\beta^{-1}$ \arnau{\cite{InteractingThermalisation}}.
This formalises the usual derivation of the canonical ensemble from the micro-canonical one in classical statistical mechanics, where the coupling energy is neglected. 
Note that the interaction strength typically scales as the boundary of the sub-system $S$. Hence, in spatial dimensions higher than one, the weak coupling assumption cannot hold true if one increases the system size. This will be the case regardless of the strength of the coupling per particle or the relative size between $S$ or $B$ \cite{ShortFarrelly12,InteractingThermalisation}. Therefore, the weak-coupling assumption is arguably inapplicable not only to realistic situations, but also to idealised systems whose constituents interact weakly.

Recently, relaxation towards equilibrium in the strong-coupling case has been addressed from the perspective of \emph{canonical typicality}. The idea is that closed non-integrable many-body systems, however described by a unitarily evolving pure state,
 are generically
expected to equilibrate \cite{Cramer_etal08,Linden2012,ShortFarrelly12,Review,CalabreseCardy,Bartsch:2009}. Such systems behave -- for the overwhelming majority of times -- as if they were described by a thermal state
when considering expectation values of local observables \cite{ETH3,InteractingThermalisation,Review}. 
The {\it eigenstate thermalisation hypothesis} \cite{ETH1,ETH2,ETH3} gives further 
substance to this expectation. This means that when a sub-system $S$ is put in contact with a bath $B$ with Hamiltonian $H_{B}$,
the equilibrium state is not (\ref{eq:Gibbstandard}), but the reduction of the global Gibbs state of $S$ and $B$
\begin{equation}\label{eq:thermalreduced}
\rho_S=\tr_B ( {\omega}(H_{SB}) ),
\end{equation}
where $H_{SB}=H_S+H_B+V$ \cite{Review}.


In this work, we incorporate these recent insights to describe equilibration to the analysis of thermodynamics beyond the weak-coupling regime. We provide rigorous 
bounds on the optimal work extraction in the presence of thermal baths whose effect is to drive 
systems to an equilibrium state of the form (\ref{eq:thermalreduced}). Our approach considers protocols of work extraction by performing 
quantum quenches on sub-systems in strong coupling with thermal baths. 

We show that the strong coupling between system and bath may induce an unavoidably irreversible component in the process and \je{we discuss} 
to what extent this results in a limitation on the optimal work extraction. We are able to quantify this deficit in terms of standard thermodynamic functions as the free energy and we show that it 
prevents one from saturating the second law of thermodynamics. Our results are completely general in the sense that they do not make use of any specific model for the description of the system or bath.

\begin{figure}
\begin{center}
\includegraphics[scale=0.23]{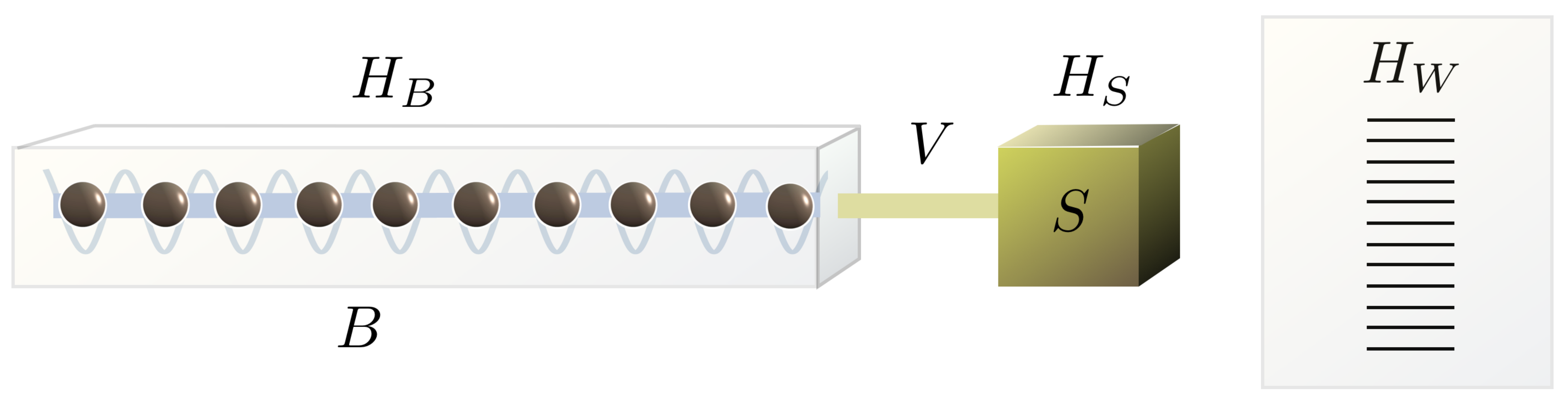}
\end{center}
\caption{\label{fig:setting}%
\textbf{Setting of the work extraction problem.} The thermal machine comprises a system $S$, a thermal bath $B$ and a battery $W$, described respectively by Hamiltonians $H_S$, $H_B$ and $H_W$. The formalism allows one to change the Hamiltonian of the system $H_S$ and to introduce an interaction $V$ between system and bath. Such Hamiltonian transformations are to be implemented by interaction with an external agent $O$ that operates the machine.
}%
\end{figure}

\section{Setting and set of operations}\label{sec:set_of_operations}
The work extraction problem requires at least the following elements:
\begin{itemize}
\item \emph{A system} $S$. This is the part of the machine upon which one has control,
i.e., it is possible to engineer its Hamiltonian $H_S$. By no means is the conservation of energy violated in this prescription.

\item \emph{A battery} $W$. {This models energy storage} and accounts for the energy supplied and extracted from the system $S$. 
It can be seen as a lifted weight. Any Hamiltonian with a suitably dense spectrum will be suitable.

\item \emph{A thermal bath} $B$. 
When the system $S$ is put into contact with the thermal bath, $S$ is assumed to thermalise in the sense of Eq.~(\ref{eq:thermalreduced}), with $H_{SB}=H_S +H_B +V$, where $V$ is the 
interaction that couples system and bath\arnau{. The interaction $V$ is assumed to be fixed and not tuneable by the operator of the machine.} No assumptions are made on the state of $SB$.
\end{itemize}
A scheme of the setting is shown in Fig.~\ref{fig:setting}.
The problem of work extraction consists of, given an initial state of $S$, an initial Hamiltonian $H^{(0)}$
and a \emph{set of operations},
transfer in expectation the maximum amount of energy from the bath to the battery.
In our case, the \emph{set of operations} are \emph{Hamiltonian transformations} and \emph{thermalisations}.

What we refer to as a {\it Hamiltonian transformation} 
a change of the Hamiltonian of the system and/or the switch on/off of the interaction $V$ between system and bath. 
Hence, at the end of each transformation, 
the Hamiltonian of  $SBW$ takes the form
\begin{equation}\label{eq:mt-generalham}
H^{(i)}=H_{S}^{(i)}+V^{(i)}+H_{B}+ H_W\, ,
\end{equation}
and $H^{(i)}$ is taken to $H^{(i+1)}$, while
%
$ V^{(i)}$ 
takes values {from} $\{0,V\}$, for $i=0,\dots, n-1$.
In order for 
this Hamiltonian transformation to be meaningful and to allow for a fair accounting of the 
work extracted, however, we require the two natural conditions to be fulfilled:
\begin{enumerate}
\item \emph{Quenches}. The reduced state on $SB$ does not change, 
\begin{equation}\label{eq:mt-quenchcondition}
\rho_{SB}^{(i)}=\rho_{SB}^{(i+1)}\, ,
\end{equation} 
{modeling the behaviour} of the system when the Hamiltonian {acting on that sub-system} is changed abruptly. 

\item \emph{Energy conservation}. The mean total energy is preserved, i.e., for each transformation
the energy change of the system (due to the change of its Hamiltonian)
has been supplied or stored by the battery
\begin{equation}\label{eq:mt-meanenergycons}
\tr ( \rho^{(i)} H^{(i)} ) = \tr ( \rho^{(i+1)} H^{(i+1)} )\, .
\end{equation}
\end{enumerate}
The first condition merely states that on the time scale of the dynamics taking place, the sudden approximation
holds true in system $SB$, or in other words, that one performs a quench. The latter is not an assumption, but rather a necessary condition for a fair account of all the energy supplied or extracted from the thermal machine $SB$. 
Note that we do not impose that the machine $SB$ is energetically isolated, which is obviously not the case since we consider time-dependent Hamiltonians that obviously do not preserve energy. Condition (\ref{eq:mt-meanenergycons}) merely states that the energy gained/lost by $SB$ is supplied or stored from the battery, which plays the role of the usual lifted weight in thermodynamics. 
The average work $\langle W \rangle_{i\to i+1}$ extracted in the quench $i\rightarrow i+1$ is
 the average energy change in the battery when the quench is performed. From Eqs.~(\ref{eq:mt-quenchcondition}) and (\ref{eq:mt-meanenergycons}) one obtains
\begin{equation}
\langle W \rangle _{i\to i+1}=\tr(\rho_{SB}^i (H_{SB}^i -H_{SB}^{i+1}))\, .
\end{equation}
This is the standard way of accounting work as the energy difference of the combined system and bath due to the 
time-dependent Hamiltonian giving rise to the evolution; in this specific case, a quench \cite{Jarzynski}.
\arnau{Note also that no assumption is made on the global state of $SBW$ and on the possible correlations between the battery $W$
and the system $S$ after implementing a quench. For the case of a unitary implementation, this issue is discussed in detail in the Appendix.}

A {\it thermalisation map} is a map that models the effect of putting the system into actual contact with the heat bath $B$ \arnau{by thermalising it}
as described in Eq.\ (\ref{eq:thermalreduced}). 
This transformation can be applied only when system and bath are interacting, that is $V^{(i)}=V$.
This family of maps is physically motivated by the realistic behaviour of evolution under generic Hamiltonians\footnote{Note that we do not require this to reflect the actual physical transformation, but the states generated 
	should for most times be locally operationally indistinguishable from those of Eq.\ (\ref{thermal}). Under reasonable assumptions, 
	this can be proven to be
	true \cite{Mueller}.}.
However,  within the abstract level of the set of operations it can be regarded simply as any completely positive map $\mathcal{T}$ acting on quantum states of $SB$ with the 
defining property 
\begin{equation}\label{thermal}
	{\tr}_{B} ({\cal T}(\rho)) = \tr_B(\omega(H_{SB})). 
\end{equation}
Applying the map $\mathcal{T}$ has no effect on the battery and hence, it does not have any work cost. 
The evolution towards equilibrium is reached by the dynamics of $SB$ alone once they are interacting, 
without having to supply or extract energy from the battery or implementing any change in the Hamiltonians of $SB$.

A sequence of such operations is called \emph{protocol}, that we denote by $\mathcal{P}$, and is specified by: (i) a list of Hamiltonians 
$\{H^{(i)}\}_{i=1}^{n-1}$ of the form (\ref{eq:mt-generalham})
and (ii) a set of instructions specifying when the thermalisation maps are realised. 
In order to avoid that the energy in the battery {originates} from a change of the system Hamiltonian, we consider protocols with the final Hamiltonian being
equal to the initial one, 
$H_S^{(n)} = H_S^{(0)}$. 
The work extracted in expectation by a protocol $\mathcal{P}$ for an initial state {$\rho_S^{(0)}$} and
an initial and final Hamiltonian $H_S^{(0)}$ is defined as the energy increase of the battery
\begin{equation}
\langle W \rangle (\mathcal{P},H^{(0)},\rho^{(0)}) :=  \tr ( (\rho_{W}^{(n)}- \rho_W^{(0)}) H_W ).
\end{equation}
Altogether, the set of operations we consider is a generalisation of the one considered in Refs.\ \cite{Abergtrullywork,Negative,EgloffRenner}, in that we
allow to change the eigenbasis of the Hamiltonian $H_{S}$. Importantly, the 
thermalisation process model is not restricted to the weak-coupling regime, {but also actually includes quantum correlations}, which alters the situation considerably.
Nonetheless, more general transformations than the ones restricted by condition (\ref{eq:mt-quenchcondition}) could be considered \cite{Anders2013}, in particular 
energy preserving unitaries that change the state of $S$, in the spirit of Refs.\ \cite{ResourceTheory,Nanomachines,Popescu2013,Popescu2013b}. 
However, Hamiltonian quenches fairly capture operational capabilities in realistic situations, rather than arbitrary unitaries,
 and are also sufficient to cover the standard weak-coupling limit 
\cite{Abergtrullywork,Nanomachines,EgloffRenner}. 
We discuss in the Supplemental Material possible ways of generalising our approach using expectation values to even further general settings and 
issues related to the
role of coherence in the battery \cite{Aberg}.

\section{Bounds on work extraction} 
Given the previous set of operations, the following theorem introduces a bound on the amount of 
work that can be extracted.

\begin{thm}[{Bounds on work extraction}]\label{mtthm:work} Given an initial state $\rho^{(0)}=\rho_{SB}^{(0)} \otimes \rho_{W}^{(0)}$ and an equal initial and final Hamiltonian $H^{(0)}$, the work that can be extracted by means of any protocol within the set of allowed operations is bounded by 
\begin{eqnarray}\label{eq:mt-boundworkmaintext}
 \langle W \rangle ({{\cal P}}, H^{(0)},\rho^{(0)}_S )  &\leq&   
 F(\tilde{\rho}_{SB},H_{SB}^{(0)} ) - F({\omega}(H_{SB}^{(0)}),H_{SB}^{(0)} ) \nonumber \\
&-& \min_{\tilde{H}_{S}} \left[ F(\tilde{\rho}_{SB},\tilde{H}_{SB} ) - F({\omega}(\tilde{H}_{SB}),\tilde{H}_{SB} ) \right]
\end{eqnarray}
where $\tilde{H}_{SB}:=\tilde{H}_S+{V}+H_B$, $\tilde{\rho}_{SB}$ is any state such that $\tr_{B}(\tilde{\rho}_{SB}^{(0)})=\tr_B (\rho_{SB}^{(0)})=  \rho_S^{(0)}$ and 
$F(\rho,H):=\tr (\rho H) -\beta^{-1}S(\rho)$
is the free energy of the state $\rho$, with respect to the Hamiltonian $H$ and inverse temperature $\beta$. Furthermore, for any initial state $\rho_{S}^{(0)}$, there exist a protocol $\mathcal{P}^{*}$ which saturates the bound.
\arnau{This optimal protocol $\mathcal{P}^{*}$ consists of a quench to the Hamiltonian $\tilde{H}_{S}$ that minimizes the difference $F(\tilde{\rho}_{SB},\tilde{H}_{SB} ) - F({\omega}(\tilde{H}_{SB}),\tilde{H}_{SB})$ followed by sequence of thermalisations and small quenches that emulates an \emph{isothermal reversible process} to come back to the initial Hamiltonian.}
\end{thm}

\begin{proof}
Note that any protocol can be expressed as a concatenation of Hamiltonian transformations and thermalisations,
and that the energy of the battery only changes in the Hamiltonian transformations.
In the first quench, the energy stored or supplied by the battery reads
\begin{equation}
\langle W \rangle_{0\to1}= \tr ( \rho_{S}^{(0)} ( H_{S}^{(0)} - H_{S}^{(1)}))=
\tr ( \tilde{\rho}_{SB}^{(0)}( H_{SB}^{(0)} - H_{SB}^{(1)})).
\end{equation}
where $\tilde{\rho}_{SB}^{(0)}$ is any state such that $\tr(\tilde{\rho}_{SB}^{(0)}) = \rho_{S}^{(0)}$.
The rest of quenches are performed after a thermalisation, hence, the work extracted from them can be written as
\begin{eqnarray}
\langle W \rangle_{i\to i+1}&=& \tr ( \rho_{SB}^{(i)} ( H_{SB}^{(i)} - H_{SB}^{(i+1)}))= \tr ( \rho_{SB}^{(i)} ( H_{S}^{(i)} - H_{S}^{(i+1)})\otimes \mathbb{I}_B ) \nonumber \\
&=&\tr (  \omega (H_{SB}^{(i)}) ( H_{S}^{(i)} - H_{S}^{(i+1)})\otimes \mathbb{I}_B ).
\end{eqnarray}
The total work extracted by a protocol is the sum of the work extracted in every Hamiltonian transformation, that is,
\begin{eqnarray}
\langle W \rangle ({\cal P}, H^{(0)},\rho^{(0)}_S ) &=& \tr \left( \tilde{\rho}_{SB}^{(0)} ( H_{SB}^{(0)} - H_{SB}^{(1)})\right)\nonumber\\
\label{eq:mt-workgeneral}&+& \sum_{i=1}^{n-1}  \tr \left(  \omega (H_{SB}^{(i)}) ( H_{SB}^{(i)} - H_{SB}^{(i+1)}) \right),
\end{eqnarray}
where $n$ is the {number of steps} of the protocol.
By using the identity $F(\omega(H),H)=\tr(H\rho)+\beta^{-1}\tr(\ln(\omega(H)\rho))$ that note that is valid for any $\rho$,
the extracted work is rewritten as
\begin{eqnarray}
\langle W \rangle ({\cal P}, H^{(0)},\rho^{(0)}_S ) &=& \tr ( \tilde{\rho}_{SB}^{(0)} ( H_{SB}^{(0)} - H_{SB}^{(1)})) \\
&+& \sum_{j=1}^{n-1} \left(F\left( \omega (H_{SB}^{(i)}) , H_{SB}^{(i)}\right)
- F\left( \omega (H_{SB}^{(j+1)}) , H_{SB}^{(j+1)}\right)\right), \nonumber \\
&-&\beta^{-1} \sum_{j=1}^{n-1} \tr\left( \ln ( \omega (H_{SB}^{(i)})) \omega (H_{SB}^{(i)}) -
\ln ( \omega (H_{SB}^{(i+1)})) \omega (H_{SB}^{(i)}) \right). \nonumber
\end{eqnarray}
After identifying the relative entropy $S(\rho \| \sigma)=\tr (\rho(\log \rho - \log \sigma))$ in the last sum of the previous equation, 
the total work becomes
\begin{eqnarray}
 \langle W \rangle ({\cal P}, H^{(0)},\rho^{(0)}_S ) &=& F(\tilde{\rho}_{SB}^{(0)},H_{SB}^{(0)} ) - F({\omega}(H_{SB}^{(0)}),H_{SB}^{(0)} ) \nonumber \\
&-&  \left( F(\tilde{\rho}_{SB}^{(0)},H^{(1)}_{SB} ) - F({\omega}(H_{SB}^{(1)}),H_{SB}^{(1)} ) \right)\nonumber \\
&-&\frac{\ln 2}{\beta} \sum_{i=1}^{n-1} S( \omega (H_{SB}^{(i+1)})\|\omega (H_{SB}^{(i)}))  . \label{eq:mt-total-work-exact}
\end{eqnarray}
Finally, the positivity of relative entropy and the inequality $F({\omega}(H),H)\leq F(\rho,H)$ complete the proof. 
The existence of a protocol $\mathcal{P}^*$ saturating the bound it is shown in \ref{sec:appendixsatwork}.
\end{proof}

Using Eq.~(\ref{eq:mt-total-work-exact}), we can identify what protocol maximises the work extracted and arbitrarily well saturates the bound.
We need to {minimize} its two negative terms, that is, (i) the second difference of free energies, and (ii) the sum of relative entropies.
The minimum of (i) is attained by choosing the first quench to the appropriate Hamiltonian $\tilde H_S$.
The term (ii) can be made arbitrarily small by performing quenches that represent a minimal change of the Hamiltonian between individual applications of thermalisation maps, 
at the expense of performing many of them.  This sequence of quenches and thermalisations precisely emulates an \emph{isothermal reversible process}.
Thus, Theorem \ref{mtthm:work} not only introduces a fundamental bound for the maximum extracted work but also
tells us what protocol arbitrarily well attains that maximum. \jen{These attainable bounds complement the findings presented in Ref.\ \cite{Allahverdyan2000}, in which the 
impact of correlations to the attainable work extraction has been considered for harmonic potentials as well as for weakly anharmonic potentials within a Fokker-Planck approach.}

Note that (\ref{eq:mt-boundworkmaintext}) contains as a particular case the well-known bounds on expected work extraction in the weak-coupling regime \cite{Popescu2013,Popescu2013b,Abergtrullywork} (see Supplemental Material). When $V$ is weak in comparison with the energy gaps of $H_{B}+H_{S}$ {to an extent that in an idealised treatment is it negligible} 
and the thermalisation process is such $\rho_{S}={\omega}(H_{S})$, then the maximum work extracted is given by the difference of free energies
\begin{equation}\label{eq:mt-workboundweakmaintext}
\langle W^{\text{wc}} \rangle ({\cal P}, H^{(0)},\rho^{(0)}_S ) \leq F(\rho_{S}^{(0)},H_{S}^{(0)})-F({\omega}(H_{S}^{(0)}),H_{S}^{(0)}).
\end{equation}
Furthermore, expression (\ref{eq:mt-boundworkmaintext}) has an insightful physical interpretation. We will show that the second line in (\eref{eq:mt-boundworkmaintext}) 
vanishes if and only if the optimal protocol is reversible. 
Otherwise, the strong coupling between system and bath induces an unavoidable dissipation in the thermalisation process that makes the protocol irreversible and limits the work 
that can be extracted.

\section{Reversibility and second law} 
We call a protocol $\mathcal{P}$ of work extraction reversible if 
${\langle W\rangle(\mathcal{P}, H}_{SB}^{(0)},\rho_{S}^{(0)})=-\langle W\rangle 
({\mathcal{P}^{-1},}
H_{SB}^{(0)},\rho_{S}^{(n)})$ where $\mathcal{P}^{-1}$ just inverts the order of the list 
and Hamiltonians $\{H^{(i)}\}_{i=1}^{n-1}$ and thermalisations of $\mathcal{P}$ and $\rho_{S}^{(n)}$ is the final state after applying $\mathcal{P}$ to $\rho_S^{(0)}$. 
\rodrigo{In other words, if $\mathcal{P}$ is a protocol that brings the system to equilibrium while extracting work, $\mathcal{P}^{-1}$ supplies work in order to bring an equilibrium state out of equilibrium. 
It is well-known that in the weak-coupling case, when the processes are optimal, $\mathcal{P^*}$ and $\mathcal{P^{*}}^{-1}$ extract/supply the same amount of work. Here we show that this is not the case in the stron-coupling case.}
One can show, by a similar argument used in the proof of Thm. \ref{mtthm:work}, that
\begin{eqnarray}\label{eq:mt-maxworkreverted}
 \langle &W& \rangle  (\mathcal{P^*}^{-1},H_{S}^{(0)},\omega_{S}(H_{SB}^{(0)}) ) \\\nonumber
&=& -F(\tilde{\rho}_{SB},H_{SB}^{(0)} ) + F({\omega}(H_{SB}^{(0)}),H_{SB}^{(0)} ) = :  \Delta F_{\text{rev}}.
\end{eqnarray}
Note that the optimal forward and reversed protocol differ exactly in the second line in (\ref{eq:mt-boundworkmaintext}), which for this reason we refer to as irreversible free-energy difference
\begin{eqnarray}\label{eq:mt-maxworkirrev}
\nonumber \Delta F_{\text{irrev}}:= -  \min_{\tilde{H}_{S}} \left[ F(\tilde{\rho}_{SB},\tilde{H}_{SB} ) - F({\omega}(\tilde{H}_{SB}),\tilde{H}_{SB} ) \right]\, .
\end{eqnarray}

Hence, even a close to being optimal protocol is surprisingly in general far from being reversible. 
The reason for the irreversibility is that when $\rho_{S}^{(0)}$ cannot be expressed as the reduced state of thermal state ${\omega}(\tilde{H}_{SB})$, then it is impossible that a 
protocol $\mathcal{P}^{-1}$ brings $\omega_{S}(H_{SB}^{(0)})$ into $\rho_{S}^{(0)}$. 
This is precisely the case when $\Delta F_{\text{irrev}}$
is not zero. Note that in the weak coupling regime this is never the case, as \emph{any} state $\rho_{S}^{(0)}$ can be expressed as a thermal state at any temperature, given that one can choose the Hamiltonian. Therefore, {in contrast} to our case, in the weak coupling case the optimal protocol is reversible.

The existence of a reversible protocol saturating the work extraction it is well-known to be related to the saturation of the second law of thermodynamics. Let us recall Clausius' theorem, that 
{in a commonly expressed variant} states that
\begin{equation}\label{eq:mt-clausiusgenmaintext}
\frac{\Delta Q}{T}=\int_i^f \frac{\delta Q}{T} \leq \Delta S =S_{f}-S_{i}\, ,
\end{equation}
where $Q$ is the heat defined as the energy lost by the bath and $S$ is the thermodynamic entropy. Most importantly, equality (saturation of second law)  holds only when the process is reversible. 
If one relates the thermodynamic with the von Neumann entropy, (\ref{eq:mt-clausiusgenmaintext}) it can be 
easily shown to imply the bounds of work extraction in the weak-coupling regime (\ref{eq:mt-workboundweakmaintext}), where indeed, the bound is saturated for reversible processes.

In the strong coupling regime, one can see that the Clausius' version of the second law (\ref{eq:mt-clausiusgenmaintext}) implies that 
\begin{equation}
\langle W \rangle ({\cal P}, H^{(0)},\rho^{(0)}_S ) \leq - \Delta F_{\text{rev}},
\end{equation}
differing from the bound of Theorem \ref{mtthm:work} precisely in $\Delta F_{\text{irrev}}$. This clarifies the role of the strong coupling in thermodynamics: The 
{\it entanglement} between system-bath induces unavoidable irreversibility that is an obstacle against 
saturating the second law of thermodynamics. It is only for particular initial states (the ones that look as reduced states of thermal states of a larger system) that reversible protocols can be implemented and the second-law can be saturated. This striking limiting effect of entanglement contrasts previous works in alternative scenarios \cite{Brunnerentanglement,Acinentanglement}, 
where entanglement was regarded {rather as an} enhancer of work extraction or power.

\section{Physical implementation in a unitary formulation}
The bounds on work extraction of our formalism coincide, in the special case of a
weak-coupling regime (\ref{eq:mt-workboundweakmaintext}), 
with previous results that employ a different set of operations based on unitary transformations \cite{Nanomachines,Popescu2013,Popescu2013b,Aberg}. There, optimal protocols employ system-bath interactions mediated by fine-tuned unitaries that differ substantially from what one would expect nature to implement generically. On the contrary, in our formalism the system-bath coupling is only required to thermalise the system following (\ref{eq:thermalreduced}), being arguably the case for most interactions. This explains the ubiquity of work extraction machines which are far from needing microscopically engineered unitaries. 
Here, in order to connect our work with this other approach, we formulate an embedding of our set of operations
 into a unitary formalism. 

\subsection{Quenches in a unitary formulation}

The standard way of describing a quench is by the sudden change of a parameter of the Hamiltonian.
By solving the time dependent Schr\"odinger equation it can be proven that if the change of such parameter is performed fast enough, 
the state of the system remains the same immediately before and after the quench.
Nevertheless, as we show next, this description of the quench has strong implications
on the properties of the system's environment.

Let us consider a two-level system $R$ with energy eigenstates $\ket{0}$ and $\ket{1}$ and 
energy levels $E_0$ and $E_1$ respectively. 
In order to perform a quench, it is also necessary to consider an \emph{environment} or \emph{battery} with Hamiltonian $H_W$
 that supplies
(stores) the lack (excess) of energy required by the quench.
The Hamiltonian of the whole setup is then $H=H_R + H_W$.
Let us now consider a \emph{unitary} process $U$ that 
performs a level transformation of the excited state $\ket{1}$ from $E_1$ to $E_1 + \Delta$.
More explicitely,
\begin{eqnarray}
\ket{0}_R\ket{0}_W \ \ &\mapsto& \ \ U\ket{0}_R\ket{0}_W=\ket{0}_R\ket{0}_W \nonumber \\
\ket{1}_R\ket{0}_W \ \ &\mapsto& \ \ U\ket{1}_R\ket{0}_W=\ket{1}_R\ket{0-\Delta}_W \, . \nonumber
\end{eqnarray}
Because of the linearity of the unitary that implements the quench, 
we can also transform an initial superposition state. For instance, the $\ket{+}_R$ becomes
\begin{equation}
\left(\ket{0}_R+\ket{1}_R\right)\ket{0}_W \ \ \mapsto \ \ \ket{0}_R\ket{0}_W + \ket{1}_R\ket{0-\Delta}_W \nonumber \, .
\end{equation}
Hence, while the initial state of the set $RW$ was a product state, the state after the quench
is entangled in a superposition for the battery of having and not having supplied energy.
This implies that \emph{it is impossible to
do quenches that leave the system unchanged if
the battery is initially in an energy eigenstate}. 
Although this conclusion seems a bit odd, it can be circumvent by having a battery with
non-distinguishable states. 

In order to clarify the above situation, 
let us think of a partition of the experiment into a system $R$, a battery $W$, and a control qubit $Q$. The total Hamiltonian of the system is
\begin{equation}\label{eq:quenchhamiltoniansimple}
H=H_R^{(0)} \otimes\mathbb{I}_W \otimes \ketbra{0}{0}_Q+H_R^{(1)} \otimes \mathbb{I}_W \otimes \ketbra{1}{1}_Q+\mathbb{I}_S \otimes H_W \otimes \mathbb{I}_Q
\end{equation}
where $H_R^{(0,1)}$ are arbitrary Hamiltonians with eigenvectors 
$\lbrace \ket{i^{(0,1)}}\rbrace_{i}$ and eigenvalues
$\{E_i^{(0,1)}\}_i$ with $i=1,\dots, d_R$, and $d_R$ is the dimension of the Hilbert space ${\cal H}_R$ of system $R$.
Note that the control qubit $Q$ dictates which is the Hamiltonian of the system.

We consider now the action of a global unitary $U$ supported on $RWQ$
on an initially uncorrelated state $\rho^{(0)}_{RWQ}=\rho^{(0)}_R \otimes \rho^{(0)}_W \otimes \ketbra{0}{0}_Q$, in a way
such that the final state can be written as 
\begin{equation}
	\rho^{(1)}_{RWQ}=U\rho^{(1)}_{RWQ}U^{\dagger}=\rho^{(1)}_{RW}\otimes \ketbra{1}{1}_Q. 
\end{equation}
In this way, according to (\ref{eq:quenchhamiltoniansimple}), the effective Hamiltonian acting on $R$ has changed from $H_R^{(0)}$ to $H_R^{(1)}$. In addition, we impose the following three natural constraints on the unitary transformation:
\begin{enumerate}
\item \emph{Energy conservation}. The unitary commutes with the Hamiltonian $[H,U]=0$.
\item \emph{Battery translational invariance}. We assume the battery to have a sufficiently dense 
equally spaced spectrum $\mathcal{W}$ (or a continuous one) with the property that 
the unitary $U$ commutes with $\id_R\otimes \Gamma_W(e)$ for all $w,w+e\in \mathcal{W}$, 
where $\Gamma_W(e)\ket{w} = \ket{w+e}$ is the tranlation operator on the battery. 
This merely reflects the invariance of the
transformation under changes of energy-origin of $H_W$ \cite{Popescu2013b}.

\item \emph{Quenches}. The unitary transformation is such that there exists an initial state of the battery $\rho_W^{(0)}$ is such that 
$\rho_R^{(1)}= \rho_R^{(0)}$ for every initial state $\rho_R^{(0)}$ and any Hamiltonian transformation 
$H_R^{(0)}\mapsto H_R^{(1)}$.
\end{enumerate}

Conditions (i) and (ii) are not present in the abstract formalism of work extraction of Section \ref{sec:set_of_operations}. We highlight that this is a desired feature of our approach: 
The general formalism that provides the above mentioned
bounds avoids as many assumptions as possible for the sake of general applicability. 
However, the particular protocol that attains the maximum fulfills {further}
conditions of physical relevance. In {particular}, 
assumption (i) allows one to extend this analysis to a single-shot work extraction, as considered in {Refs.} \cite{Abergtrullywork,EgloffRenner,Nanomachines}. We leave these {analyses} 
open for further work.
{The following} theorem shows that under {these additional} assumptions, the unitary performing the quench exists and is unique. 

\begin{thm}[{Uniqueness of unitary realisations}]\label{thm:unitary}
Consider unitary transformations such that $U\rho^{(0)} U^{\dagger}=\rho_{RW}^{(1)} \otimes \ketbra{1}{1}$ for any initial state $\rho^{(0)}=\rho^{(0)}_{RW} \otimes \ketbra{0}{0}$. The unitary that fulfils conditions (i-iii) is unique and can be written as 
\begin{eqnarray}\label{eq:unitaryquench}
\nonumber U&=&\sum_{i,j}\ket{j^{(1)}}\braket{j^{(1)}}{i^{(0)}}\bra{i^{(0)}}_R \otimes \Gamma_W (E_i^{(0)}-E_j^{(1)}) \otimes \ketbra{1}{0}_Q\\
&+&\sum_{i,j}\ket{i^{(0)}}\braket{i^{(0)}}{j^{(1)}}\bra{j^{(1)}}_R \otimes \Gamma_W (E_j^{(1)}-E_i^{(0)}) \otimes \ketbra{0}{1}_Q.
\end{eqnarray}
\end{thm}
\begin{proof}
The fact that the unitary flips the state of $Q$, implies that
\begin{equation}
U=U^{\text{on}}_{RW} \otimes \ketbra{1}{0}_Q+U^{\text{off}}_{RW}
 \otimes \ketbra{0}{1}_Q.
\end{equation}
Let us first consider the case in which the initial state $\rho^{(0)}_{RWQ}=\proj{i^{(0)}} \otimes \proj{w} \otimes \proj{0}$ is an eigenstate of $H$ with energy $E_i^{(0)}+w$.
By using condition (i), the final state after performing the unitary
is also an eigenstate of $H$ with the same energy, 
\begin{equation}
(H_R^{(1)}+H_W)U^{\text{on}}_{RW}\ket{i^{(0)}} \otimes \ket{w}= (E_i^{(0)}+w)U^{\text{on}}_{RW}\ket{i^{(0)}} \otimes \ket{w}.
\end{equation}
Hence, the state $U^{\text{on}}_{RW} \ket{i^{(0)}} \otimes \ket{w}$ is contained
in the subspace spanned by $\{\ket{j^{(1)}}\otimes \ket{E_i^{(0)}+w-E_j^{(1)}}\}_j$, that is
\begin{equation}
U^{\text{on}}_{RW}\ket{i^{(0)}} \otimes \ket{w}=\sum_{j}R_{j,i,w}\ket{j^{(1)}}\otimes \ket{\Delta_{i,j} +w},
\end{equation}
where $\Delta_{i,j} := E_i^{(0)}-E_j^{(1)}$ and $R_{j,i,w}$ are its coefficients.

By imposing condition (ii), $[U^{\text{on}}_{RW},\mathbb{I}_R \otimes \Gamma_W(E)]=0$, one gets
\begin{equation}
\sum_{j}(R_{j,i,w}-R_{j,i,w+E})\ket{j^{(1)}}\otimes \ket{\Delta_{i,j} +w+E}=0,
\end{equation}
which implies that the $R_{j,i,w}$ are independent of $w$, hence, $R_{j,i,w}=R_{j,i}$.

In order to exploit condition (iii), let us consider that the state of the battery that allows for quenches, 
i.~e~$\rho_R^{(1)}=\rho_R^{(0)}$ for any $\rho_R^{(0)}$, 
is pure and denoted by $\ket{\Psi^{(0)}}_W=\sum_{w}B_w \ket{w}$. 
We choose the global initial state to be $\rho_{RWQ}^{(0)}=\ketbra{\phi}{\phi}_R \otimes \ketbra{\Psi^{(0)}}{\Psi^{(0)}}_W \otimes \ketbra{0}{0}_Q$, with $\ket{\phi}=\sum_i c_i \ket{i^{(0)}}$. Then,
\begin{equation}\label{eq:finalpurestate}
	U^{\text{on}}_{RW} \ket{\phi}\otimes \ket{\Psi^{(0)}}
	=\sum_{j,i}c_iR_{j,i} \ket{j^{(1)}}\otimes \Gamma(\Delta_{i,j} )\ket{\Psi^{(0)}}.
\end{equation}
From Eq.\ (\ref{eq:finalpurestate}), we can compute the reduced state of $R$ 
\begin{equation}\label{eq:rho1final}
	\rho_R^{(1)} = \sum_{i,i',j,j'}c_ic_{i'}^{*}R_{j,i} R_{j',i'}^*K_{j,j'}^{i,i'} \ketbra{j^{(1)}}{j'^{(1)}},
\end{equation}
where $K_{j,j'}^{i,i'}{:= }\bra{\Psi^{(0)}}\Gamma(\Delta_{i,j} -\Delta_{i',j'} )\ket{\Psi^{(0)}}$. 
Imposing $\rho_R^{(1)}=\proj{\phi}_R$, we obtain
\begin{equation}
	R_{j,i}R^*_{j',i'}K_{j,j'}^{i,i'}=\braket{j}{i}\braket{i'}{j'} , \hspace{.8cm} \forall \ j,\ j',\ i,\ i' \, .
\end{equation}
Multiplying the previous equation by its conjugate and summing over $j$ and $j'$, one gets
\begin{equation}
\sum_{j,j'}|R_{j,i}|^2|R_{j',i'}|^2 |K_{j,j'}^{i,i'}|^2 = 1 = \sum_{j,j'}|R_{j,i}|^2|R_{j',i'}|^2,\, \forall  i,i'\, .
\end{equation}
Notice that because of the  $\Gamma(x)$ {being}
 unitary, $|K_{j,j'}^{i,i'}|\le 1$.
Hence, the only one way such that condition (iii) can hold true requires that 
\begin{equation}\label{eq:K-condition}
	|K_{j,j'}^{i,i'}|=|\bra{\Psi^{(0)}}\Gamma(\Delta_{i,j} -\Delta_{i',j'} )\ket{\Psi^{(0)}}|=1,\hspace{.5cm} \forall \, j,j',i,i'\, .
\end{equation}
This can be only satisfied for every choice of 
$E^{(1)}_{j},E^{(0)}_{i},$ (i.e., for every choice of energy levels of the initial and final Hamiltonian) if $B_w=B_{w+\Delta E}$ for every possible value of 
\begin{equation}
	\Delta E\leq \max_{j,j'}|\Delta_{i,j} -\Delta_{i',j'}|,
\end{equation}	
which in turn implies that $K_{j,j'}^{i,i'}=1$ {$\forall  j,j',i,i'$}. This, together with (\ref{eq:rho1final}), 
implies that $R_{j,i}=\braket{j^{(1)}}{i^{(0)}}$. This leads to 
\begin{equation}
U^{\text{on}}_{RW}=\sum_{i,j}\ket{j^{(1)}}\braket{j^{(1)}}{i^{(0)}}\bra{i^{(0)}}_R \otimes \Gamma_W (E_i^{(0)}-E_j^{(1)}) .
\end{equation}
This argument can be straightforwardly extended for the case of a mixed state of the battery, 
\begin{equation}
	\rho_W^{(0)}=\sum_{w}p_w\ketbra{\Psi_w^{(0)}} {\Psi_w^{(0)}}. 
\end{equation}
Also, a symmetric argument can be applied to $U^{\text{off}}_{RW}$ by considering an inverse quench $H_R^{(1)}\mapsto H_R^{(0)}$ that must leave invariant the initial state of $R$. {Altogether}, we arrive at 
Eq.\ (\ref{eq:unitaryquench}).
\end{proof}

One observation of the previous proof is that in order for the unitary to keep the state of the system invariant, 
the battery must be in an energy coherent state with a much larger uncertainty than
the operator norm of the Hamiltonians $H_R$. This is encapsulated in Eq.~(\ref{eq:K-condition}).
Hence, coherence is a resource needed to implement quenches. Contrary to the conclusions 
of Ref.\ \cite{Aberg}, coherence is destroyed due to time-evolutions of the battery with $H_W$ (see Supplemental Material). 
This suggests that the catalytic role of coherence in Ref.\ \cite{Aberg} may be a consequence of disregarding time-evolution as the mechanism for thermalisation.

\subsection{Thermalisation after the quench driven by the unitary time evolution}
Although closed quantum systems evolve according to a unitary time evolution and hence, strictly speaking, never
equilibrate, their subsystems generically do relax towards the time average state \cite{Linden2012}.
This equilibration is not exact but probabilistic, in the sense that the subsystem is very close to its time average for the overwhelming majority, but not all, times.
Furthermore, when additional assumptions are made on the bath, such as it is described by a local Hamiltonian and its state has decaying correlations, the time average state of a subsystem is the reduced of the global thermal \cite{InteractingThermalisation,Mueller}.

The previous ideas are the motivation to introduce the thermalisation map in Eq.~(\ref{eq:thermalreduced}). 
Nevertheless, in our case, there is a subtlety that has to be taken into account: the time evolution of the battery.
Note first that as during the thermalisation process 
the battery is not interacting with $SB$, the dynamics in $SB$ are independent of $W$ and
$S$ relaxes to the reduced of a thermal as it has been previously explained.  
The relevance of considering the dynamics of the battery concerns to what happens to the battery itself $W$ and in particular
to its coherence, which can lead to limitation for further quenches.

This issue is studied in detail in the Appendix.
In sum, the coherence of the battery is lost due to time evolution under its own Hamiltonian $H_W$, 
what represents an obstacle against performing further quenches in general. 
In our case, this is not a problem, since coherence is lost after a thermalisation-decoherence process 
that leaves the system-bath setting in a diagonal state in its eigenbasis and 
this allows for the implementation of further quenches.
In alternative scenarios, where quenches of systems with non-diagonal matrix elements want to be performed, 
it is a relevant question how coherence could be re-established in the battery by a certain operation -- possibly employing a device playing the role of a source of coherence. As a matter of fact, the role of {coherence} and how it should be accounted for as a resource in thermodynamics is an interesting open question that we leave open for {future} 
work. Note that the role played by coherence {in the present work} is quite different from the one {taken} in {Ref.} \cite{Aberg}. There, coherence is a catalytic resource, in the sense that it is not 
{consumed} in the protocol an can be re-used an arbitrary number of times. Our analysis points out that such catalytic role of the coherence may be only an {artefact} 
of the {specific} framework of operations considered {there}, where time-evolutions are {not taken into account}.

\section{Conclusions}
In this work we have introduced a framework to study work extraction in thermal  machines. Our formalism considers quantum Hamiltonian 
quenches as the fundamental operations and analyses the effect of strong couplings between the system and the thermal bath. 
Strikingly, system-bath entanglement seriously limits the amount of work extractable and induces irreversibility in the process, 
which in turn prevents one from saturating the second law of thermodynamics. This is relevant since
any finite-time approach to quantum thermodynamics necessarily has to take correlations and non-zero interactions into account.
Also, we introduce a formalism to embed Hamiltonian quenches into a unitary formalism. Under a set of reasonable assumptions, we show that the unitary embedding is unique and coherence is required as a resource to implement the quenches. {It should be clear that the mindset presented here can also be applied to a variety of related problems in quantum thermodynamics such as Landauer's principle
\cite{QuantitativeLandauer,MunichLandauer}, whenever
correlations are expected to be non-negligible.} Hence, this work opens new venues to understand the role of quantum effects such as entanglement and coherence in thermodynamics.

\paragraph{Acknowledgements.} {We would like to thank the EU (Q-Essence, SIQS, RAQUEL, COST, AQuS), the ERC (TAQ), the BMBF, the FQXi, and the AvH for support.}\\

\bibliographystyle{apsrev4-1}

\appendix

\section{Optimal protocol saturating the work bound}\label{sec:appendixsatwork}

Here, we show that {the bound}  (\ref{eq:mt-boundworkmaintext} )\rodrigo{can be arbitrarily well approximated. This can be most easily seen in a ``continuum limit'' of protocols, where an arbitrarily large number $n$ of operations are performed}.
The first step of the protocol \rodrigo{that arbitrarily well saturates the bound (\ref{eq:mt-boundworkmaintext})} is to perform a quench on $SBW$ from $H^{(0)}$ to $H^{(1)}=H_S^{(1)}+H_B+{V}+H_W$, where $H_{S}^{(1)}$ is the Hamiltonian that 
attains the minimum in the second term of (\ref{eq:mt-boundworkmaintext}). Applying (\ref{eq:mt-quenchcondition}) straightforwardly  one finds
\begin{equation}
\langle W \rangle ^{0,1}= \tr_{SB} ( (H_{SB}^{(0)}-H_{SB}^{(1)}) \rho_{SB}^{(0)} )=\tr_{S} ( (H_{S}^{(0)}-H_{S}^{(1)}) \rho_{S}^{(0)} ).
\end{equation}
Consider now a differentiable {parametrized curve} $H_{S}:  [0,1] \rightarrow {\cal B}(\mathcal{H}_{S})$, where ${\mathcal B}(\mathcal{H}_{S})$ denotes the 
bounded operators on the Hilbert space associated with 
$SB$. This function fulfills $H_{S}(0)=H_{S}^{(1)}$ and $H_{S}(1)=H_{S}^{(0)}$. Given an integer $n$, one defines a sequence of $n-1$ Hamiltonians as 
\begin{equation}
	H_{SB}^{(i)}:= H_{S}\left(\frac{i-1}{n-1}\right)+{V}+H_B
\end{equation}	
with $i= {1,\dots, n}$. This sequence of Hamiltonians will be used as a sequence of quenches on the equilibrated sub-system, as discussed in Section  \ref{sec:quenchesequilibrated}.
More precisely, consider a protocol in which, after the first quench from $H^{(0)}$ to $H^{(1)}$ described above, one applies a sequence state thermalisations as (\ref{thermal})  followed by quenches $H_{SB}^{(i)}\mapsto H_{SB}^{(i+1)}$ with $i= {1,\dots, n-1}$. One finds that
\begin{equation}
\langle W \rangle^{1,n}= \sum_{i=1}^{n-1} \langle W \rangle^{i,i+1} =\rodrigo{\sum_{i=1}^{n-1} }\tr \left( (H_S^{(i)}-H_S^{(i+1)})\otimes \mathbb{I}_B \: \omega (H_{SB}^{(i)}) \right).
\end{equation}
In the limit of $n$ tending to infinity, the {expected} work cost of these sequence of quenches can be written as
\begin{equation}\label{eq:workintegral}
	\lim_{n \rightarrow \infty} \langle W \rangle^{1,n}= - \int_{0}^{1} \tr \left( \frac{\partial H_{S}(\lambda)}{\partial \lambda} \otimes \mathbb{I}_B \: \omega (H_{S}(\lambda)+{V}+H_B) \right)  d \lambda.
\end{equation}
Let us denote $H_{SB}(\lambda)=H_{S}(\lambda) +{V}+H_B$, then
\begin{eqnarray}
\nonumber \frac{\partial }{\partial \lambda}&\ln  (\tr ( e^{-\beta H_{SB}(\lambda)} ) )
\\&=\frac{  \frac{\partial}{\partial \lambda}  \tr( e^{-\beta H_{SB}(\lambda)})}{ \tr( e^{-\beta H_{SB}(\lambda)})}
=\frac{    \tr( \frac{\partial}{\partial \lambda}  e^{-\beta H_{SB}(\lambda)})}{ \tr( e^{-\beta H_{SB}(\lambda)})} \nonumber\\
\label{eq:wilcoxformula}&= \frac{\tr ( \int_{0}^{1} e^{-\alpha \beta H_{SB}(\lambda)} \frac{\partial}{\partial \lambda}( - \beta H_{SB}(\lambda)) \: e^{-(1-\alpha) \beta H_{SB}(\lambda)} )\: \text{d} \alpha  }   { \tr( e^{-\beta H_{SB}(\lambda)}) }  \\
&= \frac{ \int_{0}^{1} \tr ( e^{-\alpha \beta H_{SB}(\lambda)} \frac{\partial}{\partial \lambda}( - \beta H_{SB}(\lambda)) \: e^{-(1-\alpha) \beta H_{SB}(\lambda)}  ) \: \text{d} \alpha}   { \tr( e^{-\beta H_{SB}(\lambda)}) }
\nonumber \\
&=  \frac{  \tr (   \frac{\partial }{\partial \lambda} (- \beta (H_{SB}(\lambda)) \: e^{-\beta H_{SB}(\lambda)} ) }   { \tr( e^{-\beta H_{SB}(\lambda)}) }
\\\label{eq:wilcoxfinal} &=-\beta \tr \left( \frac{\text{d}H_S(\lambda)}{\text{d}\lambda} \otimes \mathbb{I}_B \: \omega (H_{S}(\lambda)+{V}+H_B) \right) , 
\end{eqnarray}
where Eq.\ (\ref{eq:wilcoxformula}) follows from Wilcox formula for matrix exponential derivatives \cite{Wilcoxformula}. By combining Eq.\ (\ref{eq:wilcoxfinal}) with Eq.\ 
(\ref{eq:workintegral}) and $F({\omega}(H),H)=-  \ln (\tr (e^{-\beta H}))/\beta$, one finds 
\begin{equation}
\lim_{n \rightarrow \infty} \langle W \rangle^{1,n}= F ( {\omega}(H_{SB}^{(1)}),H_{SB}^{(1)}) )- F ( {\omega}(H_{SB}^{(0)}),H_{SB}^{(0)}) ),
\end{equation}
hence, the total work extracted in the process is 
\begin{eqnarray}
\rodrigo{\lim_{n \rightarrow \infty}} \langle W \rangle^{1,n} 
(H^{(0)}, \rho_{S}^{(0)})&=&\tr_{S} ( (H_{S}^{(0)}-H_{S}^{(1)}) \rho_{S}^{(0)} )+ F ( {\omega}(H_{SB}^{(1)}),H_{SB}^{(1)}) )\\&-& F ( {\omega}(H_{SB}^{(0)}),H_{SB}^{(0)}) )\nonumber\\
\nonumber &=& F(\tilde{\rho}_{SB},H_{SB}^{(0)} ) - F({\omega}(H_{SB}^{(0)}),H_{SB}^{(0)} ) \\
&-& F(\tilde{\rho}_{SB},H_{SB}^{(1)} ) - F({\omega}(H_{SB}^{(1)}),H_{SB}^{(1)} )\\\label{eq:realisationfinal}
\nonumber &=&F(\tilde{\rho}_{SB},H_{SB}^{(0)} ) - F({\omega}(H_{SB}^{(0)}),H_{SB}^{(0)} ) 
\\&-& \min_{\tilde{H}_{S}} \left[ F(\tilde{\rho}_{SB},\tilde{H}_{SB} ) - F({\omega}(\tilde{H}_{SB}),\tilde{H}_{SB} ) \right],\label{eq:realisationfinal2}
\end{eqnarray}
where Eq.\ (\ref{eq:realisationfinal}) follows from calculations equivalent to the ones from the proof of Thm. \ref{thm:optimalprotocol} and (\ref{eq:realisationfinal2}) is derived 
from the choice of $H_{SB}^{(1)}$. 

\section{Bounds on work extraction in the weak coupling limit} 
In the weak couplin limit the effect of the bath is to drive the system $S$ to an {equilibrium} state ${\omega}(H_{S})$ 
that is Gibbs, since
\begin{equation}
	{\omega}(H_{SB})\approx {\omega}(H_{S})\otimes \omega (H_{B}), 
\end{equation}
then one can take $\tilde{\rho}_{SB}={\omega}(\tilde{H}_{S})\otimes \omega (H_{B})$ and $\tilde{H}_S=-\ln (\rho_{S}^{(0)})/\beta$ and a simple calculation shows that in that case $\Delta F_{\text{irrev}}=0$ and $\Delta F_{\text{rev}}$ does not depend on $H_B$, so that
\begin{equation}\label{eq:workweakcoupling}
\max_{\mathcal{P}} \langle W \rangle ( \mathcal{P},H^{(0)},\rho^{(0)}_S )  \leq F(\rho_{S}^{(0)},H_{S}^{(0)} ) - F({\omega}(H_{S}^{(0)}),H_{S}^{(0)} ).
\end{equation}
Let us now comment on the precise role of the two terms $\Delta F_{\text{rev}}$ and  $\Delta F_{\text{irrev}}$ appearing in the bound Eq.\ 
(\ref{eq:mt-boundworkmaintext}), as defined in Eq.\ (\ref{eq:mt-maxworkreverted}) and Eq.\ (\ref{eq:mt-maxworkirrev}). 
Consider first a modified scenario in which $SB$ are treated as larger working medium that we denote by $S'$. In such scenario, one has full control over the Hamiltonian of $S'$, that is $H_{SB}$, and furthermore, that $S'$ can be driven to the Gibbs equilibrium state ${\omega}(H_{SB})$ -- this may be achieved by weak-coupling with a bath $B'$ that interacts with $SB$. In this case, similar analysis to the one leading to (\ref{eq:workweakcoupling}) shows that the maximum work extracted is precisely $\Delta F_{\text{rev}}$. Hence, $\Delta F_{\text{irrev}}$ should be understood as a work penalty due to our lack of control over $H_{B}$, and therefore, through expression (\ref{thermal}), over the {equilibrium} state of $S$.


\section{Coherence as a resource for quenches}\label{sec:coherenceforquenches}
From the proof of Theorem \ref{thm:unitary} it is clear that one needs a specific initial state of the battery {$\rho_W(0)= \ket{\Psi^{(0)}}\bra{\Psi(0)}$} 
in order to guarantee that the state of $R$ is not altered by the change of Hamiltonian. This is encapsulated in the following condition
\begin{equation}\label{eq:conditionbattery}
K_{j,j'}^{i,i'}:= \bra{\Psi^{(0)}}\Gamma(\Delta_{i,j} -\Delta_{i',j'})\ket{\Psi^{(0)}} =1\:\: \forall \: i,i',j,j'.
\end{equation}
This {condition}
can be achieved by employing an initial state {vector} of the battery $\ket{\Psi^{(0)}}=\ket{\Psi_{\sqcap}}$ with
\begin{equation}\label{eq:defstatecoherent}
\ket{\Psi_{\sqcap}}=\frac{1}{{N(E^{(0)}_W,\Delta)}} \sum_{w=E^{(0)}_W}^{E^{(0)}_W+\Delta} \ket{w},
\end{equation}
where $N(E^{(0)}_W,\Delta)$ is the number of states with energy between $E^{(0)}_W$ and $E^{(0)}_W +\Delta$, according to the discretisation chosen. Then,
\begin{equation}
K_{j,j'}^{i,i'}=\frac{\Delta-|E_{j'}^{(1)}-E_j^{(1)}+E_{i}^{(0)}-E_{i'}^{(0)}|}{\Delta}.
\end{equation}
Therefore, by assuming 
\begin{equation}\label{eq:condition3}
\max_{i,i',j',j}\frac{|E_{j'}^{(1)}-E_j^{(1)}+E_{i}^{(0)}-E_{i'}^{(0)}|}{\Delta} \leq \epsilon
\end{equation}
by taking $\Delta$ so that $\epsilon>0$ is arbitrarily small we obtain $K_{j,j'}$ arbitrarily close to one.

Let us now analyze how the state of the battery is changed after the quench from $H_R^{(0)}$ to $H_R^{(1)}$. Starting from an initial state
\begin{equation}
	\rho_{RWQ}^{(0)}=\rho_R^{(0)} \otimes  \ketbra{\Psi_{\sqcap}}{\Psi_{\sqcap}} _W\otimes \ketbra{0}{0}_Q,
\end{equation}
applying (\ref{eq:unitaryquench}) one finds that
\begin{eqnarray}\label{eq:stateafterquench}
\rho_{RW}^{(1)}=\sum_{i,i',j,j'} \braket{j^{(1)}  }{i^{(0)}  } \bra{i^{(0)} } &\rho_R^{(0)}&  \ket{i'^{(0)} }
  \braket{i'^{(0)}  }{j'^{(1)}  }  \ket{j^{(1)}}\nonumber \\
&\times&\bra{j'^{(1)}  }_R  \Gamma(\Delta_{i,j})\ketbra{\Psi_{\sqcap}}{\Psi_{\sqcap}} \Gamma^{\dagger}(\Delta_{i',j'}).
\end{eqnarray}
Let us define the expected \emph{work} extracted in the process, as the mean-energy difference between the initial and the final state of the battery. Then,
\begin{eqnarray}
	\langle W \rangle &:= &\Trace_W ( H_W (\rho_W^{(1)}-\rho_W^{(0)} ) )\nonumber \\
	&=& \Trace_W ( H_W \sum_{i,i',j} \braket{j^{(1)}  }{i^{(0)}  } \bra{i^{(0)} } \rho_R^{(0)}  \ket{i'^{(0)} } \braket{i'^{(0)}  }{j^{(1)}  }   \Gamma(\Delta_{i,j})\ketbra{\Psi_{\sqcap}}{\Psi_{\sqcap}} \Gamma^{\dagger}(\Delta_{i',j}) ) \nonumber \\
		&-&\Trace_W ( H_W \ketbra{\Psi_{\sqcap}}{\Psi_{\sqcap}} )\nonumber \\	
	\nonumber&=&\sum_{i,i',j} \braket{j^{(1)}  }{i^{(0)}  } \bra{i^{(0)} } \rho_R^{(0)}  \ket{i'^{(0)} } \braket{i'^{(0)}  }{j^{(1)}  }  \frac{1}{N(E^{(0)}_W,\Delta)} \\
           &\times& \sum_{w,w'=E^{(0)}_W}^{E^{(0)}_W+\Delta}
     \sum_{e=-\infty}^{\infty}e\braket{e}{\Delta_{i,j}+w}\braket{\Delta_{i',j}+w'}{e} 	-\Trace_W ( H_W \ketbra{\Psi_{\sqcap}}{\Psi_{\sqcap}} ).
\end{eqnarray}
From condition Eq.\ (\ref{eq:condition3}) in the limit $\epsilon \rightarrow 0$, we get
\begin{eqnarray}
	\label{eq:workbattery1}\langle W \rangle &= &\sum_{i,i',j} \braket{j^{(1)}  }{i^{(0)}  } \bra{i^{(0)} } \rho_R^{(0)}  \ket{i'^{(0)} } \braket{i'^{(0)}  }{j^{(1)}  } \frac{1}{N(E^{(0)}_W,\Delta)}  \sum_{w{=}E^{(0)}_W}^{E^{(0)}_W+\Delta} (\Delta_{i,j}+w)
	\nonumber \\
	&-&\Trace_W ( H_W \ketbra{\Psi_{\sqcap}}{\Psi_{\sqcap}} )\\
	\label{eq:workbattery2}&=&\sum_{i}\bra{i^{(0)} } \rho_R^{(0)}  \ket{i^{(0)} }E_i^{(0)} -  \sum_{i,i',j} \braket{j^{(1)}  }{i^{(0)}  } \bra{i^{(0)} } \rho_R^{(0)}  \ket{i'^{(0)} } \braket{i'^{(0)}  }{j^{(1)}  } E_j^{(1)}\\
	\label{eq:workbattery3}&=&\Trace_R ( (H_R^{(0)} -H_R^{(1)} ) \rho_R^{(0)} ),
\end{eqnarray}
where Eq.\ (\ref{eq:workbattery2}) follows from the fact that
\begin{equation}
	\Trace_W ( H_W \ketbra{\Psi_{\sqcap}}{\Psi_{\sqcap}} )=\frac{\rodrigo{1}}{N(E^{(0)}_W,\Delta)}  \sum_{w{=}E^{(0)}_W}^{E^{(0)}_W }w. 
\end{equation}
In short, Eq.\ (\ref{eq:workbattery3}) formalises the intuition that the expected energy provided (stored) by the battery is just the expected energy gained (lost) by the system $R$ upon the quench is applied. Indeed, (\ref{eq:workbattery3}) can be derived {straightforwardly} from {the}
conservation of expected energy of $RWQ$ and the fact that the state of $R$ does not change. However, we derive it explicitly for consistency check, and also as an illustrative example of how to deal with similar calculations that appear in further sections.

 \section{Quenches with classical battery}
 
The unitary (\ref{eq:unitaryquench}) is the transformation that changes the effective Hamiltonian acting on $R$, while leaving the state invariant. As shown in previous sections, a sufficiently coherent initial state of the battery is necessary to perform such transformation. Here, we study what is the effect of the unitary (\ref{eq:unitaryquench}) if the initial state of the battery is a classical state. We will show how the state of $R$ is indeed disturbed when one implements that change of Hamiltonian and how it relates with the work extracted by the battery in such process.
Let us consider an initial state 
\begin{equation}
	\rho_{RWQ}^{(0)}=\rho_R^{(0)} \otimes \ketbra{0}{0}_W \otimes \ketbra{0}{0}_Q.
\end{equation}	
We choose the battery to be initialised in the state $\ketbra{0}{0}_W$ for ease of notation, but the extension to other pure initial states, or convex mixtures of eigenstates of $H_W$ is straightforward. The final state of $RW$ after the quench is
\begin{equation}
\rho_{RW}^{(1)}=\sum_{i,i',j,j'}\braket{j^{(1)}}{i^{(0)}}\bra{i^{(0)}}\rho_R^{(0)}\ket{i'^{(0)}}\braket{i'^{(0)}}{j'^{(1)}}\: \ketbra{j^{(1)}}{j'^{(1)}}_R \otimes \ketbra{\Delta_{i,j}}{\Delta_{i',j'}}_W.
\end{equation}
The final state of the system $R$ will depend heavily on the degeneracies of both $H_R^{(0)}$ and $H_R^{(1)}$, and also on the degeneracies of the energy differences $\Delta_{i,j}$. Let us, assume that the initial state is diagonal in the eigenbasis of $H_R^{(0)}$. That is 
\begin{equation}
	\rho_R^{(0)}= \sum_i \bra{i^{(0)}}\rho_R^{(0)}\ket{i^{(0)}} \proj{i^{(0)}}.
\end{equation}	
In this case
\begin{eqnarray}
\langle W \rangle &=& \Trace_W ( H_W \rho_{W}^{(1)} )\nonumber \\
&=& \Trace_W ( H_W  \sum_{i,i',j}\braket{j^{(1)}}{i^{(0)}}\bra{i^{(0)}}\rho_R^{(0)}\ket{i'^{(0)}}\braket{i'^{(0)}}{j^{(1)}}\: \ketbra{\Delta_{i,j}}{\Delta_{i,j}}_W 
) \nonumber \\
&=&  \sum_{i,j}\braket{j^{(1)}}{i^{(0)}}\bra{i^{(0)}}\rho_R^{(0)}\ket{i^{(0)}}\braket{i^{(0)}}{j^{(1)}}\: (E_i^{(0)}-E_j^{(0)}) \nonumber \\
&=& \sum_i \bra{i^{(0)}}\rho_R^{(0)}\ket{i^{(0)}}E_i^{(0)}- \sum_{j} \bra{j^{(1)}} ( \sum_i \bra{i^{(0)}}\rho_R^{(0)}\ket{i^{(0)}} \proj{i^{(0)}} )\ket{j^{(1)}} E_j^{(0)} \nonumber\\
\label{eq:workclass}&=& \Trace_R ( (H_R^{(0)}-H_R^{(1)})\rho_R^{(0)}).
\end{eqnarray}
Note that condition (\ref{eq:workclass}) is a necessary condition for the set of operations of the work-extracting protocol. Therefore, for classical states of the battery, the quench formalism only can be applied to extract work if the initial state $\rho_R^{(0)}$ is diagonal.



\section{Motivation for taking the reduced of a Gibbs state as equilibrium state}\label{sec:equilibriumstate}

{We now turn to the discussion of the physical mechanism that renders the thermalisation map plausible. Indeed, it captures what one naturally would expect when bringing a small body into contact
with a heat bath. In the above axiomatic approach we again leave the mechanism unspecified; here, we will explain why the above framework is indeed very meaningful and physically plausible.
In in one way or the other, the evolution to an equilibrium Gibbs state is essential in the functioning of any thermal machine.}
{The precise setting considered, however, varies within recent approaches to the study of thermal machines.}
{Within} the formalism {presented in} Refs.\ \cite{EgloffRenner,Abergtrullywork} a classical system is put in contact with a thermal bath. The system is classical in the sense that it is described a state 
$\sigma_S=\sum_j \sigma_j \proj{j}$ {that is diagonal at all times}, where $\{\ket{j}\}$ {denotes} the eigenvectors of
 a {Hamiltonian} {$H_S$} in a given state of the process. The evolution towards the Gibbs state in this formalism states that the probability distribution 
 is modified and eventually reaches an equilibrium state given by 
\begin{equation}\label{hs}
	{\omega(H_S)=\frac{e^{-\beta H_S}}{\Trace ( e^{-\beta H_S})}.}
\end{equation}	
An {alternative} approach that has been employed successfully to the study of thermal machines is {rooted in the framework} of {quantum mechanical}
resource theories. {Within such resource theories, the allowed operations have to be specified, as well as the ``free resources''.}
{Here, the role of the ``free resources'' is assumed by Gibbs states with respect to some Hamiltonians and inverse temperature} 
\cite{Nanomachines,Popescu2013,Aberg}. The work extraction process is described by a global unitary {transformation} on the {sub-systems prepared in} {Gibbs} states, {a} system $S$,
{as well as} a battery. {Within such an approach, actual evolution generated by Hamiltonians is not made explicit, and neither is the dynamics leading to equilibration and thermalisation.}
Nevertheless, the {allowed resource states are Gibbs states, which are, even if this is not made explicit, of course the result of some equilibration process, possibly involving a larger system.
Again, the Gibbs states considered a resource are of the form as in Eq.\ (\ref{hs}), with the role of $H_S$ taken over by the Hamiltonians of the sub-systems constituting the resources.
In this sense, both approaches are similar in that they crucially rely on Gibbs states of Hamiltonians that are entirely non-interacting with any other part of the system.}

However, {such an assumption can be a rather implausible one in a number of situations. In fact, this assumption is often excessively restrictive, whenever sub-systems thermalizing are not
entirely decoupled from their environment.}
Gibbs states {have been shown to emerge in systems small systems very weakly interacting with a large physical body under a number of standard assumptions on the density of states
\cite{InteractingThermalisation}.} 
Such an approach is {meaningful} in a regime in which 
\begin{equation}
	\| V \|\ll {\beta}^{-1}. 
\end{equation}
As $\|  V \|$ is in general extensive, however, {and} ${1}/{\beta}$ is an intensive quantity, such {a} regime {is only meaningful} in spin chains or restricted {forms} 
of interactions \cite{ShortFarrelly12}. {One can surely hope for better bounds that also extend to wider range of physical situations. However, in 
systems with non-negligible interactions, one would not even expect the above to be a good approximation: One would not expect sub-systems to be well described by
Gibbs states with respect to the Hamiltonians of the respective sub-systems. Thermalisation then naturally rather means that the reduced states becomes locally indistinguishable from
the reduced state of a global Gibbs state (see, e.g., Refs.\ \cite{Kliesch2013,AcinIntensivetemp,Acinlocaltemp})}. {Specifically, if one thinks of a local Hamiltonian $H_{SB}$ that can for any
region of the lattice $S$ and its complement $B$ be decomposed into  
\begin{equation}
	H_{SB}= H_S+H_B+V, 
\end{equation}
one would not expect $\EE_t( \rho_{S} (t))$ to be close to 
$\omega(H_S)$: Surely the interaction captured by $V$ will alter $\EE_t( \rho_{S} (t))$ significantly.
In the light of these considerations, it seems inadequate to ground the analysis of thermal machines on the existence of resource systems prepared in equilibrium Gibbs states in situations
in which  interactions can not be considered negligible.}

{Still, Gibbs states of course play an important role in  the description of typical equilibrium reduced states of many-body systems, only that it is the Gibbs states of larger systems that have to be taken into
account.} Consider {again} a system $S$ and a system $B$ that embodies a large number of degrees of freedom, evolving under the Hamiltonian $H_{SB}=H_S+H_B +{V}$, where no assumption is made about the strength of the interaction {term} ${V}$. For typical {local} interactions and initial states, {and in the absence of local conserved quantities}, one expects that 
\begin{equation}
\EE_t \left\|  \rho_{S}(t)  - \Trace_B \left( \frac{e^{-\beta H_{SB}}}{\Trace ( e^{-\beta H_{SB}})} \right)  \right\|_1   \ll 1,
\end{equation}
{where $\EE_t$ denotes the expectation in time.}
{This is a consequence of the sub-system being close in trace-norm for most times if the so-called effective dimension is large \cite{Linden2012,ShortFarrelly12},
and the expectation that the time averaged state reduced to $S$ is indistinguishable from $\Trace_B(\omega(H))$.}
That is, {again}, sub-systems {are} for {most times} {expected to be operationally} indistinguishable from the reduced state of {the} Gibbs state on a the {entire} 
system $SB$. This is precisely the kind of evolution towards equilibrium on which we base our description of thermal machines. 

\begin{assum}[{Thermalisation in the presence of interactions}]\label{ass:thermal} {Consider} a system composed of a sub-system $S$, a bath $B$ and a battery $W$.
{This assumption states that one} can place an interaction ${V}$ between the sub-system and the bath {such} that the evolution under the Hamiltonian 
\begin{equation}
H_{SBW}=H_S+H_B+{V}+H_W=H_{SB}+H_W
\end{equation}
for any initial state $\rho_{SBW}(t=0)$ and after {an appropriately chosen} relaxation time $\tau$ fulfills 
\begin{equation}\label{eq:thermalT}
	 \rho_S  (t=\tau)= \Trace_{B} \left( \frac{e^{-\beta H_{SB}}}{\Trace ( e^{-\beta H_{SB}})} \right).
\end{equation}
\end{assum}

{The time $\tau>0$ may well be chosen probabilistically based on a suitable measure, 
and the statement can be weakened to be true with overwhelming probability. Surely, one would expect $\rho_S$ to be locally close to the
reduction of the time average for the overwhelming proportion of, but not all, times \cite{Linden2012,ShortFarrelly12}. However, precise error bounds for the equilibration time beyond free models \cite{Cramer_etal08}
are still an arena of active research. For the purposes of the present work, therefore, we will take the pragmatic attitude that appropriate times $\tau$ 
can be taken such that 
the natural condition Eq.\ (\ref{eq:thermalT}) holds true. In the framework of our formalism, this assumption will be taken as a physically plausible assumption, and no attempts are being
made as to deriving bounds to equilibration times.
}

Treating the thermalisation map (\ref{thermal})
as the result of an {actual} time evolution compels one to apply also a time evolution to the battery. As we discuss in   \ref{sec:coherenceloss} this will result in a loss of the coherence of the battery, which 
{renders it} in general impossible to perform further Hamiltonian quenches on $SB$. However, in realistic situations, the thermal machine $SBW$ can be assumed to be weakly interacting with a surrounding environment. This will effectively produce decoherence -- that is, damping the off diagonal terms {in} the Hamiltonian eigenbasis \cite{Linden2012}. As there is no interaction between $SB$ and $W$, 
however, both are weakly interacting with a local environment, decoherence is {is expected to be most effective} 
on the product eigenbasis of $H_{SB}+H_{W}$. This effect, as we show in (\ref{eq:workclass}) allows one to perform further quenches without the need of coherence.  

\begin{assum}[{Decoherence map}]\label{ass:decoherence}
{Consider} a system composed of a system $SB$ and a battery $W$, {equipped with} a non-interacting Hamiltonian
\begin{equation}
H_{SBW}=H_{SB}+H_W.
\end{equation}
{This assumption states that}
the evolution induced by the interaction of $SBW$ with a {suitable natural} environment $E$ is equivalent {with the application of} a decoherence map $\mathcal{E}$ described by
\begin{equation}
\mathcal{E}(\rho_{SBW})=\sum_{i,w} (\rho_{SBW})_{i,i}^{w,w}\proj{i}\otimes \proj{w},
\end{equation}
where $\rho_{SBW}=\sum_{i,i',w,w'}  (\rho_{SBW})_{i,i'}^{w,w'} \ketbra{i}{i'} \otimes \ketbra{w}{w'}$, $H_{SB}=\sum_{i} E_i \proj{i}$ and \rodrigo{$H_W=\sum_{w}w \proj{w}$. }
\end{assum}

\section{Quenches on equilibrated systems}\label{sec:quenchesequilibrated}

{We will now turn to analyzing the} formalism of quenches described in Section \ref{sec:quenches} when the change of Hamiltonian is implemented on a sub-system $S$ in contact with a thermal bath $B$. Consider an initial global state $\rho_{SBW}^{(0)}(t=0)$   and an initial Hamiltonian for the thermal machine $H^{(0)}_S+H_B+{V}+H_W$. We {then allow} this system {to}
equilibrate according to this Hamiltonian, so that the evolution fulfills 
{Assumptions} \ref{ass:thermal} and \ref{ass:decoherence}. Hence, at large enough time $\tau$ the state can be written as,
\begin{equation}
	\rho_{SBW}^{(0)}(\tau)= \sum_{i,w} (\rho^{(0)}_{SBW}(0))_i^w \: \proj{i^{(0)}}\otimes \proj{w}
\end{equation}
where $(\rho^{(0)}_{SBW}(0))_i^w =\bra{i^{(0)}}\otimes \bra{w}\rho^{(0)}_{SBW}(0)\ket{i^{(0)}} \otimes \ket{w}$ and 
\begin{equation}
	H_{SB}^{(0)}:= H^{(0)}_S+H_B+{V}=\sum_{i} E_i^{(0)} \proj{i^{(0)}}. 
\end{equation}
Also, the equilibrated state fulfills
\begin{equation}\label{eq:stateisthermal}
	\rho_{S}^{(0)}(\tau)=\Trace_{B} \left( \frac{e^{-\beta H_{SB}}}{\Trace ( e^{-\beta H_{SB}})} \right).
\end{equation}
At time $\tau$ we perform a quench $H^{(0)}_S+H_B+{V} \mapsto H^{(1)}_S+H_B+{V}=\sum_{i} E_i^{(1)} \proj{i^{(1)}}$. The state after the quench $\rho_{SBW}^{(1)}(t=T)$ {satisfies}
\begin{equation}
\rho_{SBW}^{(1)}(\tau)= \sum_{i,w} (\rho^{(0)}_{SBW}(0))_i^w \: U^{\text{on}}_{RW} \proj{i^{(0)}}\otimes \proj{w} U^{\text{on}\dagger}_{RW},
\end{equation}
where $U^{\text{on}}$ is the quench unitary as defined in Eq.\ (\ref{eq:unitaryquench}).
Hence, the work extracted at the battery is
\begin{eqnarray}
\langle W \rangle &=& \Trace_{W}(H_W \rho_{W}^{(1)}(\tau) )- \Trace_{W}(H_W \rho_{W}^{(0)}(\tau) )\nonumber \\
\nonumber &=& \sum_{i,w} (\rho^{(0)}_{SBW}(0))_i^w \Trace_{W}(H_W \Trace_{SB}( U^{\text{on}}_{RW} \proj{i^{(0)}}\otimes \proj{w} U^{\text{on}\dagger}_{RW})) \nonumber\\
&-& \sum_{i,w} (\rho^{(0)}_{SBW}(0))_i^w \Trace_{W}(H_W   \Trace_{SB}(\proj{i^{(0)}}\otimes \proj{w} ) \nonumber\\
 \label{eq:quenchthermal1}&=& \sum_{i,w} (\rho^{(0)}_{SBW}(0))_i^w ( \Trace_{SB}((H_{SB}^{(0)}-H_{SB}^{(1)})\proj{i^{(0)}} ) +w ) \\
&-& \sum_{i,w} (\rho^{(0)}_{SBW}(0))_i^w w \nonumber\\
&=& \Trace_{SB}\left((H_{SB}^{(0)}-H_{SB}^{(1)}) \rho^{(0)}_{SB}(\tau)\right) \nonumber\\
&=& \label{eq:quenchthermal2}\Trace_{SB}\left(\left((H_{S}^{(0)}-H_{S}^{(1)})\otimes \mathbb{I}_B\right) \frac{e^{-\beta H_{SB}}}{\Trace ( e^{-\beta H_{SB}})} \right),
\end{eqnarray}
where Eq.\ (\ref{eq:quenchthermal1}) follows from Eq.\ (\ref{eq:workclass}), and Eq.\ (\ref{eq:quenchthermal2}) from Eq.\ (\ref{eq:stateisthermal}).

\section{Physical protocol saturating {the} work extraction bound}\label{sec:physicalprot}

We now combine the {statements} of Eqs.\ (\ref{eq:workbattery3}) and (\ref{eq:quenchthermal2}) in order 
to show that the work extraction protocol as defined in Section \ref{sec:set_of_operations} can be implemented.

\begin{cor}[{Physical implementation in a unitary framework}]\label{thm:optimalprotocol} 
Given {an} initial {state} of the form $\rho^{(0)}=\rho^{(0)}_{SB}\otimes \proj{\Psi_{\sqcap}}_W \otimes \proj{0}_Q$, with $\ket{\Psi_{\sqcap}}_W$ as defined in (\ref{eq:defstatecoherent}), and an {arbitrary} initial Hamiltonian $H^{(0)}$. {Assuming the validity of} Assumptions \ref{ass:thermal} and \ref{ass:decoherence}, any protocol $\mathcal{P}$ can be implemented with a unitary transformation acting on a sytem composed of the thermal machine $SBW$, the control qubit $Q$ and an environment $E$. 
\end{cor}
\begin{proof}
{This statement} follows straightforwardly {from} Assumptions \ref{ass:thermal} and \ref{ass:decoherence}, and Eq.\ (\ref{eq:workbattery3}) and (\ref{eq:quenchthermal2}). 
Given the initial state $\rho^{(0)}=\rho^{(0)}_{SB}\otimes \proj{\Psi_{\sqcap}}_W \otimes \proj{0}_Q$, (\ref{eq:workbattery3}) shows that the {quench} unitary (\ref{eq:unitaryquench}) performs the first Hamiltonian transformation of an arbitrary protocol $\mathcal{P}$ -- before the first state thermalisation --  so that it fulfills condition (\ref{eq:mt-quenchcondition}). Then, the unitary evolution under of the composed system $SBWQE$ {satisfying} Assumptions \ref{ass:thermal} and \ref{ass:decoherence} {results in} further quenches {fulfilling} (\ref{eq:quenchthermal2}), which in turn implies that it fulfills (\ref{eq:mt-quenchcondition}) when applied on thermalised states as in Eq.\ (\ref{thermal}).
\end{proof}

\section{Coherence in the battery and time evolution}\label{sec:coherenceloss}
As we {have discussed} in \ref{sec:coherenceforquenches}, a coherent state of the battery allows one to perform a Hamiltonian quench. This can be easily seen from (\ref{eq:stateafterquench}), if one applies a quench to an initial state of the form 
\begin{equation}
	\rho_{RWQ}^{(0)}=\rho_R^{(0)} \otimes  \ketbra{\Psi_{\sqcap}}{\Psi_{\sqcap}} _W\otimes \ketbra{0}{0}_Q
\end{equation}
 -- $R$ plays the role of system plus bath -- the reduced final state on $R$ does not change, that is 
 \begin{equation}
 	\rho_{R}^{(1)}=\rho_{R}^{(0)}.
\end{equation}	
Let us suppose that now we let the system $RW$ undergo a time-evolution under the Hamiltonian $H^{(1)}_R+H_W$ -- this is precisely what one {does} if $R$ {embodies both} 
a {system $S$ and a bath}, and the time-evolution is {intended} to drive $\rho_{R}^{(1)}$ towards a thermalised state of the form (\ref{thermal}). How does this time evolution {affect} 
the coherence in the state of the battery? {Is the battery still coherent} so that it can perform further quenches? Here we show that this is not the case. Coherence is a resource that gets lost under 
{such a} time evolution. To see this, let us {compute} the time-evolved state after time $t$ of $\rho_{R}^{(1)}$ which is {given by}
\begin{eqnarray}\label{eq:evolvedrw}
&&\rho_{RW}^{(1)}(t)=\sum_{i,i',j,j'} \braket{j^{(1)}  }{i^{(0)}  } \bra{i^{(0)} } \rho_R^{(0)}  \ket{i'^{(0)} } \braket{i'^{(0)}  }{j'^{(1)}  }  e^{-i(E_j^{(1)}-E_{j'}^{(1)})t}\ketbra{j^{(1)} }{j'^{(1)}  }_R \nonumber \\
\nonumber &\otimes & \frac{1}{N(E^{(0)}_W,\Delta)} \sum_{w=E^{(0)}_W}^{E^{(0)}_W+\Delta} e^{-i(\Delta_{i,j}+w)t}\ket{\Delta_{i,j}+w}\sum_{w'=E^{(0)}_W}^{E^{(0)}_W+\Delta}e^{i(\Delta_{i',j'}+w)t} \bra{\Delta_{i',j'}+w'}.
\end{eqnarray}
From this equation one can {straightforwardly}, but tediously, {conclude} that 
\begin{equation}
	\rho_{R}^{(1)}(t)=e^{-iH_R^{(1)}t}\rho_{R}^{(0)}e^{iH_R^{(1)}t},
\end{equation}	
that is, as one should expect, the initial state evolved under $H_R^{(1)}$ {at} time $t$. Now, if one intends to perform further quenches on this state -- that is, a unitary of the form (\ref{eq:unitaryquench}) 
{changing} $H_{R}^{(1)}$ to $H_{R}^{(2)}$ without altering the state on $R$ -- one finds that this is not possible, because the state of the battery has been changed by {the}
evolution under $H_W$ and it no longer serves as a coherent resource fulfilling (\ref{eq:conditionbattery}). This can be shown by a {tedious} calculation applying the unitary (\ref{eq:unitaryquench}) {on}
(\ref{eq:evolvedrw}). To avoid such {a} calculation and {merely} grasp the intuition behind the mechanism, note that the state vector
\begin{equation}
	\ket{\Psi_{\sqcap}(t)}=\frac{1}{{N(E^{(0)}_W,\Delta)^{1/2}}} \sum_{w=E^{(0)}_W}^{E^{(0)}_W+\Delta} e^{-i(\Delta_{i,j}+w)t}\ket{\Delta_{i,j}+w}
\end{equation}
no longer fulfills (\ref{eq:conditionbattery}) when a new quench from $H_R^{(1)}$ to $H_{R}^{(2)}$ --with energy gaps $\Delta_{i,j}^{(2)}$-- is applied. Indeed, it is easy to see that for {most times} $t$ 
\begin{equation}
\bra{\Psi_{\sqcap}(t)}\Gamma(\Delta^{(2)}_{i,j} -\Delta^{(2)}_{i',j'})\ket{\Psi_{\sqcap}(t)} \approx 0.
\end{equation}
In other words, the coherence of the battery is lost due to time evolution under its own Hamiltonian $H_W$, and this {is an obstacle against performing} further quenches in general. In the specific protocol leading to Corollary \ref{thm:optimalprotocol}, further quenches can be applied because coherence is no longer needed after the decoherence map specified in Assumption \ref{ass:decoherence}
{has been applied}. 
We {expect} this decoherence map to {reasonably} represent {plausible and realistic} physical situations. However, {it should be clear} that alternative protocols in which, for instance, coherence is re-established in the battery by a certain operation -- possibly employing a device playing the role of a source of coherence -- are {also}
of great interest. As a matter of fact, the role of {coherence} and how it should be accounted for as a resource in thermodynamics is an interesting open question that we leave open for {future} 
work. Note that the role played by coherence {in the present work} is quite different from the one {taken} in {Ref.} \cite{Aberg}. There, coherence is a catalytic resource, in the sense that it is not 
{consumed} 
in the protocol an can be re-used an arbitrary number of times. Our analysis points out that such catalytic role of the coherence may be only an {artefact} 
of the {specific} framework of operations considered {there}, where time-evolutions are {not taken into account}.

\section{Spread of energy probability distribution and single-shot considerations}

\label{sec:discussion-single-shot}
{As far as work extraction is concerned}, in {our} work, we follow the approach of, e.g.,  Ref.\ \cite{Popescu2013b} and consider average work extraction. 
{Our results hence apply to the expected work for individual systems}. Note we do not have to assume at any point, --similarly to as in Ref.\ \cite{Popescu2013}-- that we process $N$ copies collectively in order to obtain (\ref{eq:mt-boundworkmaintext}). Due to linearity of the work extraction process, it is implied by {a basic argument of typicality} 
that when processing $N$ copies, the total work extracted per copy will be {essentially}
deterministic in the limit of large $N$ -- the variance increases with $\sqrt{N}$ and the total work with $N$. However, it is still of interest to analyse the spread of the probability distribution of the energy in the battery for a single copy. This is relevant {with} generalisations to single-shot work extraction in the spirit of {Refs.} \cite{Abergtrullywork,Nanomachines,EgloffRenner} {in mind}. 
Note that such analysis is out of place within the abstract formalism defined in Section  \ref{sec:set_of_operations}: The operations just preserve 
{the expected} energy, thus transformations reducing arbitrarily the spread of the energy of the battery are allowed, similarly as in the formalism defined in {Ref.} \cite{Popescu2013b}. Nonetheless, note that the unitary implementation of the protocol of Corollary \ref{thm:unitary} does preserve the probability distribution of the {entire}
machine $SBW$. This is the case because (i) the unitary defined in Theorem \ref{thm:unitary} does not only {preserve} 
the mean total energy, but {it} also commutes with the total Hamiltonian and (ii) the dephasing map employed when the system relaxes to an equilibrium state, as defined in Assumption \ref{ass:decoherence}, by definition preserves the probability distribution of energies of  $SBW$. Therefore, one could restrict the set of operations defined in Section  \ref{sec:set_of_operations}, by substituting the assumption of mean-energy conservation for a conservation of the probability distribution of total energy, and a protocol saturating (\ref{eq:mt-boundworkmaintext}) would still be attainable. In conclusion, the formalism itself, {in contrast} 
to the one in {Ref.} \cite{Popescu2013b}, can be easily modified to account for a possible generalisation in therms of single-shot work extraction.

Nevertheless, {there is} another issue {that} prevents one from applying {straightforwardly} the findings of {Refs.} 
\cite{Abergtrullywork,Nanomachines,EgloffRenner}: This is the impossibility of performing quenches with deterministic classical states of the battery. As detailed in Section
 \ref{sec:coherenceforquenches}, one needs to employ initial state vectors of the battery $\ket{\psi_{\sqcap}}_W$. Therefore, the initial probability distribution of energies of the battery is already 
{``infinitely spread out''}. 
As discussed in Ref.\ \cite{Aberg}, a distinction between ordered work -- as the single-shot work extraction -- and disordered work would need to take into account the energy carrier -- in this case the battery -- and how the initially spread distribution of the battery is affected by the protocol. We leave this as an interesting open question that lies out of the scope of this {work}.

%

\section{Typicality of irreversibility and second law}\label{sec:irreversibility}
{We now turn to the discussion of the typicality of irreversibility and the relationship to an instance of a second law.}
The equivalence between optimality and reversibility in work extraction protocols has been widely known in the context of phenomenological thermodynamics,
{the} analysis of the Carnot engine being the most seminal example. More {generally}, Clausius' theorem states that overall heat flow vanishes over \emph{all} reversible cyclic processes. That is,  
\begin{equation}
\oint_{\text{rev}} \frac{\delta Q}{T}=0
\end{equation}
where $\delta Q$ is the inexact differential of the heat $Q$ and $T$ is the temperature. This motivates the definition of the entropy state function as $dS :=  {\delta Q}/{T}$, 
{$T$ taking the role of the integrating factor.}
Furthermore, Clausius' inequality establishes that for general processes -- not necessarily reversible or cyclic -- it is {true that}
\begin{equation}\label{eq:clausiusgen}
\frac{\Delta Q}{T}=\int_i^f \frac{\delta Q}{T} \leq \Delta S =S_{f}-S_{i},
\end{equation}
where {equality} holds in the reversible case.

{This} theorem is formulated within the framework of phenomenological thermodynamics. {However}, {similar} expressions can be shown to hold {within a statistical mindset}
with {the} von Neumann entropy {taking over} the role of thermodynamic entropy \cite{Anders2013}. Indeed, in the weak-coupling {setting}, 
it is not difficult to show that (\ref{eq:clausiusgen}) is indeed equivalent to the {bounds on expected work extraction}
in terms of free-energy difference, and also that optimal work extraction processes are reversible. 

To see this, consider a protocol of work extraction by Hamiltonian quenches as defined in Section  \ref{sec:set_of_operations}. In the weak-coupling regime, the state thermalisation map (\ref{thermal}) 
is replaced by {$\rho_{S}^{(i)}\mapsto {\omega}(H_{S}^{(i)})$}. Equivalent with Eq.\ (\ref{eq:mt-workgeneral}), the {expected}
work extracted in a general protocol in the weak-coupling limit is {given by}
\begin{equation}
\langle W^{\text{wc}} \rangle ( \mathcal{P},H_{S}^{(0)},\rho_{S}^{(0)})= \Trace ( \rho_{S}^{(0)}(H_{S}^{(0)}-H_{S}^{(1)}) ) + \sum_{i=1}^{n-1} \Trace  ({\omega}(H_{S}^{(i)}) ( H_{S}^{(i)}- H_{S}^{(i+1)}) ),
\end{equation}
{which}, recalling (\ref{eq:workweakcoupling}), fulfills
\begin{equation}\label{eq:workboundweak}
\max_{\mathcal{P}} \langle W^{\text{wc}} \rangle ( \mathcal{P},H_{S}^{(0)},\rho_{S}^{(0)}) \leq F( \rho_{S}^{(0)},H_{S}^{(0)} )-F( {\omega}(H_{S}^{(0)}),H_{S}^{(0)}).
\end{equation}
{Equality is achieved here} by a reversible process. Now let us see that {a} similar conclusion can be reached from (\ref{eq:clausiusgen}). If we define the heat flow $\Delta Q$ as the energy lost by the bath -- or {equivalently}, the energy gained by the system in the state thermalisation process -- one can see that 
\begin{eqnarray}
\nonumber\langle \Delta Q^{\text{wc}} \rangle ( \mathcal{P},H_{S}^{(0)},\rho_{S}^{(0)})&=& \Trace ( ({\omega}(H_{S}^{(1)})-\rho_{S}^{(0)})H_{S}^{(1)} ) \\&+& \sum_{i=1}^{n-1} \Trace (({\omega}(H_{S}^{(i+1)})-{\omega}(H_{S}^{(i)} ))H_{S}^{(i+1)}) )\\
\nonumber &=& -\Trace ( \rho_{S}^{(0)}H_{S}^{(1)} ) +\sum_{i=1}^{n-1} \Trace ({\omega}(H_{S}^{(i)}) ( H_{S}^{(i)}- H_{S}^{(i+1)}) )\\&+& \Trace ({\omega}(H_{S}^{(n)}) H_{S}^{(n)})\\
&=&\langle W^{\text{wc}} \rangle ( \mathcal{P},H_{S}^{(0)},\rho_{S}^{(0)})+ \langle \Delta E \rangle_S
\end{eqnarray}
where 
\begin{equation}
	\langle \Delta E \rangle_S:= \Trace ({\omega}(H_{S}^{(n)}) H_{S}^{(n)})-\Trace ( \rho_{S}^{(0)} H_{S}^{(0)}) 
\end{equation}
is the {expected} energy difference between the initial and final state. Therefore, identifying $\langle \Delta Q^{\text{wc}} \rangle ( \mathcal{P},H_{S}^{(0)},\rho_{S}^{(0)})$ with the heat flow, and the von Neumann entropy with the thermodynamic entropy in (\ref{eq:clausiusgen}), one obtains
\begin{eqnarray}
\langle W^{\text{wc}} \rangle ( {\mathcal{P}},H_{S}^{(0)},\rho_{S}^{(0)})&\leq& T \Delta S_S -\langle \Delta E \rangle_S\nonumber \\
&=&F( \rho_{S}^{(0)},H_{S}^{(0)} )-F( {\omega}(H_{S}^{(0)}),H_{S}^{(0)}),
\end{eqnarray}
where, according to Clausius' theorem, {equality} {again} holds when the process is reversible. This equivalence between Clausius' theorem and the work extraction bounds means that indeed 
(\ref{eq:workboundweak})
{may} 
be understood as an alternative formulation of the second law of thermodynamics {applied to expectation values}. 
Also, the fact that there exist an optimal reversible protocol saturating (\ref{eq:workboundweak})
is to be understood as saturation of the second law.

Let us now investigate the situation where the interaction between bath and system is not necessarily weak and the thermalisation map is of the form (\ref{thermal}). As anticipated in Section \ref{sec:set_of_operations}, in general the coupling between bath and systems {prevents} one {from saturating} the second law {in the form stated above} 
and to perform reversible processes. 

The first difference when analyzing the strong-coupling case, is that the very definition of heat is problematic. In a system evolving from $\rho_{SB}^{(i)}$ to $\rho_{SB}^{(i)}$ {equipped}
with the Hamiltonian $H_{S}+H_{B}+V_{SB}$, it is {not quite clear} how much energy is lost by the bath -- this is the canonical definition of heat -- because the energy {contribution}
of the interaction is not negligible, and it is not obvious which part corresponds to the bath and to the system. To motivate a way of circumvent this problem, let us consider an specific example. Let us partition the bath $B$ into two regions $B_b$  {(the buffer)} and $B_r$ {(the reservoir)}. The buffer represents the region of the bath that is surrounding they system $S$ and the reservoir is the region that is not directly in contact with $S$. Let us suppose that $B_b$ and $B_r$ are weakly coupled, so that the {operator} 
norm $V_{B_b,B_r}$ is much smaller {than} 
the energy gaps of their respective Hamiltonians. This would be the case if, for instance, $S$ and $B_b$ are {parts} 
of a conducting material, and $B_r$ is just a {surrounding} gas that interacts weakly with $B_b$. For such setup the equilibration towards equilibrium of $S$ will fulfill,
\begin{equation}\label{eq:thermalbuffer}
	\rho_{S,B_b}=\omega (H_{SB_b}),
\end{equation}
where $H_{SB_b}=H_S+H_{B_b}+V_{SB_b}$, and $V_{SB_b}$ is an arbitrarily strong interaction that has only support in $B_b$ ({but} not in $B_r$). In this case, weak interaction between $B_b$ and $B_r$ establishes a clear cut that allows on to {unambiguously} define  the energy that was lost by the the reservoir $B_r$ -- {in contrast} to the energy that {has} flown from $B_b$ to $S$ that is ambiguous due to the strong coupling in $V_{B_b,S}$.  Hence, the definition of heat can {be} 
made unambiguous as the energy lost by the reservoir $B_r$, or equivalently, the energy gained by $SB_b$. Taking this as the definition of heat, one obtains
\begin{eqnarray}
	\langle \Delta Q \rangle ( \mathcal{P},H_{SB_b}^{(0)},\rho_{SB_b}^{(0)})&=& 
	\Trace \left( \left(\omega(H_{SB_b}^{(1)})-\rho_{SB_b}^{(0)}\right)H_{SB_b}^{(1)} \right) \\&+& \sum_{i=1}^{n-1} \Trace \left (\left (\omega(H_{SB_b}^{(i	+1)})-\omega(H_{SB_b}^{(i)} )\right)H_{SB_b}^{(i+1)} \right)\nonumber\\
	\nonumber &=& -\Trace \left( \rho_{SB_b}^{(0)}H_{SB_b}^{(1)} \right) +\sum_{i=1}^{n-1} \Trace \left (\omega(H_{SB_b}^{(i)}) ( H_{SB_b}^{(i)}- H_{SB_b}^{(i+1)}) \right)\\&+& \Trace \left (\omega(H_{SB_b}^{(n)}) H_{SB_b}^{(n)} \right)\nonumber\\
	&=&\langle W \rangle ( \mathcal{P},H_{SB_b}^{(0)},\rho_{SB_b}^{(0)})+ \langle \Delta E \rangle_{SB_b},
\end{eqnarray}
where 
\begin{equation}
	\langle \Delta E \rangle_{SB_b}:=\Trace \left (\omega(H_{SB_b}^{(n)}) H_{SB_b}^{(n)} \right)- \Trace \left( \rho_{SB_b}^{(0)}H_{SB_b}^{(0)} \right). 
\end{equation}
Using (\ref{eq:clausiusgen}) and identifying $\Delta S = S ( \omega(H_{SB_b}^{(n)}) )- S ( \rho_{SB_b}^{(0)})$, one obtains
\begin{eqnarray}
	\langle W \rangle ( \mathcal{P},H_{SB_b}^{(0)},\rho_{SB_b}^{(0)})&\leq & T \Delta S_{SB_b} -\langle \Delta E \rangle _{SB_b}\nonumber\\
	&=&F\left( \rho_{SB_b}^{(0)},H_{SB_b}^{(0)} \right)-F\left( \omega(H_{SB_b}^{(0)}),H_{SB_b}^{(0)}\right)\nonumber\\
	\label{eq:frevappears}&=&- \Delta F_{\text{rev}},
\end{eqnarray}
where (\ref{eq:frevappears}) is {a consequence of} Theorem \ref{thm:optimalprotocol}
and {taking} $\tilde{\rho}_{SB_b}= \rho_{SB_b}^{(0)}$.

Lastly, in the case of the thermalisation map of the form (\ref{thermal}), where no assumption is made about a cut between {the} buffer and {the}
reservoir, the {entire} bath has to be considered the buffer $B_b$ and the reservoir is not present. Then, in analogy to (\ref{eq:thermalbuffer}), if we strengthen condition (\ref{thermal}) by assuming that the equilibrium state fulfills
\begin{equation}
\rho_{SB}=\omega(H_{SB}),
\end{equation}
one can define heat unambiguously as the energy gained by the whole machine -- which {vanishes} by {an argument based on the} conservation of energy. Indeed, we find 
\begin{eqnarray}\label{eq:ceroheat}
\langle \Delta Q \rangle ( \mathcal{P},H_{SB}^{(0)},\rho_{SB}^{(0)})
&=&\langle W \rangle ( \mathcal{P},H_{S}^{(0)},\rho_{S}^{(0)})+ \langle \Delta E \rangle_{SB}=0,
\end{eqnarray}
where last equality follows simply from {expected} energy conservation. Therefore, in a scenario {based on a} a thermalisation map of the kind {considered in Eq.} (\ref{thermal}), the second law can be written simply as
\begin{equation}\label{eq:secondlaw}
0 \leq \Delta S_{SB}
\end{equation}
where {equality} is fulfilled by {a} reversible process. This together with (\ref{eq:ceroheat}) gives again
\begin{eqnarray}
	\langle W \rangle ( \mathcal{P},H_{SB}^{(0)},\rho_{SB_b}^{(0)})&\leq & T \Delta S_{SB} -\langle \Delta E \rangle _{SB}\nonumber\\
	&=&F\left( \rho_{SB}^{(0)},H_{SB}^{(0)} \right)-F\left( \omega(H_{SB}^{(0)}),H_{SB}^{(0)}\right)\nonumber\\
	\label{eq:frevappears2}&=&- \Delta F_{\text{rev}},
\end{eqnarray}
where the equality is {satisfied} by reversible protocols of work extraction. The bound of {Theorem} \ref{thm:optimalprotocol} thus \je{usually}
imposes a limitation, quantifiable by $\Delta F_{\text{irrev}}$, {against} 
{saturating}  the second law of thermodynamics (\ref{eq:secondlaw}). The reason, as the very formulation of the second law by Clausius' theorem already {takes into account},
 is that the process is not reversible. This can be easily seen from Eqs.\ (\ref{eq:realisationfinal}, \ref{eq:realisationfinal2}). The optimal protocol specifies a Hamiltonian $H_{SB}^{(1)}$, and 
 {parametrized curve} of Hamiltonians describing a trajectory from $H_{SB}^{(1)}$ to $H_{SB}^{(0)}$. Now one can reverse the protocol, that is, given $\mathcal{P}$ by $\{H_{SB}^{(i)}\}_{i=1}^{n-1}$ and $\textbf{k}$, we define the inverse protocol $\mathcal{P}^{-1}$ by $\{H_{SB}^{(n-i)}\}_{i=1}^{n-1}$ and $\textbf{k}^{-1}(i):=  \textbf{k}(n-i)$, and a simple calculation shows
\begin{eqnarray}
\nonumber \langle &W \rangle  (\mathcal{P}^{(-1)},H_{SB}^0, \omega_S(H_{SB}^{(0)}))\\
 &= \Trace \left( \omega(H_{SB}^{(1)})(H_{S}^{(1)}- H_{S}^{(0)}) \right ) + F\left( \omega (H_{SB}^{(0)}),H_{SB}^{(0)}\right) -F\left( \omega (H_{SB}^{(1)}),H_{SB}^{(1)}\right) \nonumber\\
&=F\left( \omega (H_{SB}^{(0)}),H_{SB}^{(0)}\right) - F\left( \omega (H_{SB}^{(1)}),H_{SB}^{(0)}\right) \nonumber\\
&=\langle W \rangle (\mathcal{P},H_{SB}^0, \rho_{SB}^{(0)}))+ \Delta F_{\text{irrev}}.
\end{eqnarray}
That is, the work difference between the optimal protocol and its reversed protocol is precisely $\Delta F_{\text{irrev}}$. This quantity is exactly the amount by which the work extraction bounds differ from the maximum ones allowed by the second law {stated in the form} $0\leq \Delta S_{SB}$. Altogether, this suggests that Theorem \ref{thm:optimalprotocol} {may} 
be viewed as a generalisation of the second law of thermodynamics which accounts for strong couplings and the unavoidable irreversibility that it induces.

The irreversibility of the optimal process may result in {a}
tension with Theorem \ref{thm:optimalprotocol}, where it is shown that a global unitary evolution performs the optimal protocol, and therefore it must be reversible. This apparent paradox {is resolved} 
by noting that {being} 
reversible at the level of abstract protocols -- that is, as we define $\mathcal{P}^{-1}$ -- is not equivalent with {being} 
reversible in the sense of time-reversed implementation. 
Note that the time-reversed evolution can take {equilibrium} states to states out of equilibrium, however, 
a reversed protocol in the sense of $\mathcal{P}^{-1}$ does not allow for such passages from {equilibrium} to 
non-{equilibrium} states. This is precisely the case, for example, in the first step of the optimal protocol detailed in the proof of {Theorem}
\ref{thm:optimalprotocol}. There, the initial Hamiltonian $H_{SB}^{(0)}$ is quenched to $H_{SB}^{(1)}$, and then the state of $SB$ is driven to equilibrium, so that 
$\rho_{S}^{(0)} \mapsto \omega_{S}(H_{SB}^{(1)})$. Clearly, this equilibration is {eventually} due to {some} 
unitary evolution of the composed system $SB$, and indeed could be in principle reversed if one had control over the exact time that we waited until 
\begin{equation}
	\rho_{S}^{(0)}\rightarrow \omega_{S}(H_{SB}^{(1)})
\end{equation}
{has converged}.
However, at {the} 
abstract level {mainly considered here}, where work extraction protocols $\mathcal{P}$ are {being} defined, 
the protocols {neither explicitly} take {time} into account {nor any} other dynamical analysis of the state thermalisation. 
Therefore, a reversed protocol of the previous example would just {amount to}
a quench from $H_{SB}^{(1)}$ to $H_{SB}^{(0)}$ on the state $\omega_{S}(H_{SB}^{(1)})$. The use of the abstract map (\ref{thermal}) is grounded precisely in typicality arguments, as explained in \ref{sec:equilibriumstate}. {In} other words, the irreversibility exhibited by the optimal protocols, should be understood also as a feature of typicality: 
Given the precise times that one {has} 
waited in each equilibration process, ${\bf t}=(\tau_1,\tau_2,\dots, \tau_l)$, for {most} times, with {overwhelmingly} high 
probability, the optimal protocol extracts $- (\Delta F_{\text{rev}} - \Delta F_{\text{irrev}} )$. If one applies the reversed protocol, with {suitable} times for equilibration, for 
{most} times and all initial states, with {overwhelming} probability, the work extracted in the inverse protocol would be $-\Delta F_{\rm rev}$. Therefore, the optimal protocol is typically irreversible.\bigskip

\end{document}